\pdfoutput=1
\documentclass[aps,prresearch,reprint,groupedaddress,longbibliography]{revtex4-2}

\usepackage{amsmath,amssymb,amsthm}
\usepackage{graphicx}
\usepackage{xcolor}
\usepackage{bm}
\usepackage[colorlinks=true,linkcolor=blue,citecolor=blue,urlcolor=blue]{hyperref}

\newcommand{\E}{\mathbb{E}}
\newcommand{\Prob}{\mathbb{P}}
\newcommand{\hist}{\mathcal{F}}
\newcommand{\nnode}{n_{\mathrm{node}}}
\newcommand{\gnode}{\gamma_{\mathrm{n}}}
\newcommand{\gedge}{\gamma_{\mathrm{e}}}
\newcommand{\loglik}{\ell}
\newcommand{\logT}{\loglik_{\mathrm{T}}}
\newcommand{\logS}{\loglik_{\mathrm{S}}}
\newcommand{\thT}{\bm{\theta}_{\mathrm{T}}}
\newcommand{\thS}{\bm{\theta}_{\mathrm{S}}}
\newcommand{\thW}{\bm{\theta}_{\mathrm{W}}}
\newcommand{\logW}{\loglik_{\mathrm{W}}}
\newcommand{\Cov}{\operatorname{Cov}}
\newcommand{\Ent}{\mathrm{H}}

\newtheorem{proposition}{Proposition}
\newtheorem{lemma}{Lemma}
\newtheorem{corollary}{Corollary}
\newtheorem{remark}{Remark}

\newcommand{\synN}{40}
\newcommand{\synT}{600}
\newcommand{\synEvents}{1553}
\newcommand{\synTrueNnode}{0.50}
\newcommand{\synTrueKappa}{20.0}
\newcommand{\synObsNnode}{0.491}
\newcommand{\synObsNnodeSE}{0.021}
\newcommand{\synObsKappa}{19.45}
\newcommand{\synObsKappaSE}{1.43}
\newcommand{\synLatNnode}{0.343}
\newcommand{\synLatKappa}{24.29}
\newcommand{\synLatCorr}{-0.60}
\newcommand{\synSEInflationNnode}{1.08}
\newcommand{\synSEInflationKappa}{1.34}

\newcommand{\sweepReps}{30}          
\newcommand{\sweepEventsLow}{53}       
\newcommand{\sweepEventsHigh}{14{,}032} 
\newcommand{\sweepCells}{64}
\newcommand{\sweepCellsOK}{31}
\newcommand{\sweepEventsMinOK}{194}

\newcommand{\gofPowerLawKS}{0.24}
\newcommand{\gofPowerLawACF}{-0.6}
\newcommand{\gofPowerLawEvents}{1834}

\newcommand{\eStepBiasSyn}{-30\%}       
\newcommand{\sumExpSchoolSingle}{0.91}
\newcommand{\sumExpSchoolMix}{0.95}
\newcommand{\sumExpSchoolGain}{1200}
\newcommand{\sumExpHospSingle}{0.99}
\newcommand{\sumExpHospMix}{0.97}
\newcommand{\sumExpHospGain}{175}

\newcommand{\hsNodes}{327}
\newcommand{\hsNodesActive}{302}
\newcommand{\hsRawSamples}{188508}
\newcommand{\hsSpanHours}{101}
\newcommand{\hsDays}{5}

\newcommand{\hsDayZeroContacts}{12472}
\newcommand{\hsDayZeroNnode}{0.863}
\newcommand{\hsDayZeroNnodeSE}{0.009}
\newcommand{\hsDayZeroKappa}{56.1}
\newcommand{\hsDayZeroKappaSE}{1.4}
\newcommand{\hsDayZeroNodeMemory}{446}     
\newcommand{\hsDayZeroTieMemory}{28011}    
\newcommand{\hsDayZeroCorr}{+0.24}

\newcommand{\sfNodes}{403}
\newcommand{\sfRawSamples}{70261}
\newcommand{\sfSpanHours}{31.8}

\newcommand{\hoNodes}{75}
\newcommand{\hoRawSamples}{32424}
\newcommand{\hoSpanHours}{96.5}

\newcommand{\sessionGap}{60}           
\newcommand{\samplingGrid}{20}         
\newcommand{\minContacts}{5}

\newcommand{\perdayHsNnode}{0.808 \pm 0.128}
\newcommand{\perdayHsKappa}{67.8 \pm 12.3}
\newcommand{\perdaySfNnode}{0.812 \pm 0.085}
\newcommand{\perdaySfKappa}{18.2 \pm 7.4}
\newcommand{\perdayHoNnode}{0.950 \pm 0.010}
\newcommand{\perdayHoKappa}{8.6 \pm 0.6}
\newcommand{\gofHsCalibP}{0}                 
\newcommand{\gofHsCalibPSel}{0}             
\newcommand{\schoolShortDayNnode}{0.580}

\newcommand{\fitUnderFold}{2.0}         
\newcommand{\posConEvNode}{600}         
\newcommand{\posConNlo}{0.20}\newcommand{\posConNloFit}{0.204}
\newcommand{\posConNhi}{0.80}\newcommand{\posConNhiFit}{0.806}
\newcommand{\posConKlo}{5}\newcommand{\posConKloFit}{4.9}
\newcommand{\posConKhi}{50}\newcommand{\posConKhiFit}{49.9}
\newcommand{\groupAtomDy}{0.034}       
\newcommand{\groupAtomGp}{0.003}       
\newcommand{\groupZdy}{4.3}            
\newcommand{\groupZgp}{4.6}            
\newcommand{\groupNnode}{0.92}         

\newcommand{\gofKSHighSchool}{0.093}
\newcommand{\gofACFHighSchool}{-4.1}
\newcommand{\gofMeanTau}{0.923}
\newcommand{\gofVarTau}{1.025}
\newcommand{\gofZeroAtom}{9.3\%}
\newcommand{\gofTiedFraction}{8.8\%}
\newcommand{\gofKSHighSchoolDetied}{0.063}
\newcommand{\gofACFHighSchoolDetied}{-4.8}
\newcommand{\gofHsWinContacts}{7{,}192}
\newcommand{\gofHsWinN}{302}
\newcommand{\gofHsWinHours}{5.0}

\newcommand{\gofEmContacts}{10{,}000}
\newcommand{\gofEmN}{653}

\newcommand{\gofBootReps}{100}

\newcommand{\gofHsNullMedianKS}{0.0234}
\newcommand{\gofHoCalibPTiming}{0}
\newcommand{\gofHoCalibPSelection}{0.17}
\newcommand{\gofHoNullMedianKS}{0.0207}
\newcommand{\gofSfCalibPTiming}{0}
\newcommand{\gofSfCalibPSelection}{0.79}
\newcommand{\gofSfNullMedianKS}{0.0211}
\newcommand{\planeMhospital}{+0.09}
\newcommand{\planeMconf}{+0.15}
\newcommand{\planeMschool}{+0.20}

\newcommand{\snnM}{-0.04}
\newcommand{\snnRefM}{-0.28}
\newcommand{\snnRefAcfZ}{-28}
\newcommand{\fssElo}{67}
\newcommand{\fssEhi}{3386}
\newcommand{\fssDropSmallE}{0.25}

\newcommand{\fssDropMidLo}{0.15}
\newcommand{\fssDropMidHi}{0.16}
\newcommand{\fssWidthSlope}{+0.11}
\newcommand{\fssWidthSlopeLo}{+0.05}
\newcommand{\fssWidthSlopeHi}{+0.19}
\newcommand{\fssRatEightLo}{0.58}
\newcommand{\fssRatEightHi}{0.77}
\newcommand{\fssRatRecLo}{0.88}
\newcommand{\fssRatRecHi}{0.93}

\newcommand{\deTieHoR}{0.92}

\newcommand{\deTieHsR}{0.73}
\newcommand{\deTieHsLo}{0.24}
\newcommand{\deTieHsHi}{0.98}
\newcommand{\deTieSfR}{0.80}
\newcommand{\deTieSfLo}{0.38}
\newcommand{\deTieSfHi}{0.99}
\newcommand{\rawHsWinR}{0.97}
\newcommand{\rawHsWinLo}{0.94}
\newcommand{\rawHsWinHi}{1.02}
\newcommand{\tieReuseSchool}{0.77}
\newcommand{\tieUsesPerEdge}{4.3}
\newcommand{\epiBeta}{0.15}
\newcommand{\epiNode}{0.55}
\newcommand{\epiLink}{0.65}
\newcommand{\epiSisNode}{0.43}
\newcommand{\epiSisLink}{0.58}
\newcommand{\pairShufNull}{0.20}
\newcommand{\pairShufNullLo}{0.00}
\newcommand{\pairShufNullHi}{1.72}
\newcommand{\isiShufNull}{1.01}
\newcommand{\nullReps}{60}
\newcommand{\snnObsNnode}{0.909}
\newcommand{\snnLatNnode}{0.806}
\newcommand{\emObsNnode}{0.361}
\newcommand{\emLatNnode}{0.300}
\newcommand{\eStepBiasSnn}{$-$11\%}

\newcommand{\confSynN}{100}

\newcommand{\confSynSmooth}{-0.62}

\newcommand{\kapInflSo}{28\%}

\newcommand{\kapInflWb}{45\%}

\begin{document}

\title{Fundamental limits to identifying node and tie memory in temporal networks: marginal artefacts and spreading dynamics}

\author{Michele Tizzani}
\email{micti@dtu.dk}
\affiliation{Technical University of Denmark (DTU), Kongens Lyngby, Denmark}

\begin{abstract}
Temporal-network models place the memory seen in contact data either on nodes, through
self-excitation (branching ratio $\nnode$), or on ties, through reinforcement ($\kappa$). The two
choices reproduce the same coarse statistics but predict different epidemics. We ask whether
contact data can tell them apart and find two limits. The first is attribution. When the
initiator of every event is recorded, the likelihood factorises into a timing block containing
$\nnode$ and a selection block containing $\kappa$, so the two are estimated independently. When
it is not, as in proximity data, a missing-information term couples them, with a size set by the
entropy of the initiator posterior. The second issue is interpretation. Even where the mechanisms
separate, the fitted $\nnode$ might not accurately measure memory and it may be set by the
distribution of waiting times between events. Shuffling each node's inter-event order, which keeps that distribution, typically changes $\hat{n}_{\mathrm{node}}$ by less
than $10\%$; this is a bias that longer records and a flexible scale-free kernel do not remove.
Whether node memory is present must therefore be decided by model-free statistics, which find only weak order memory in the proximity records, far less than the fitted $\nnode\approx0.9$ would imply. Tie reinforcement remains recoverable, although a latent initiator inflates it by \kapInflSo{}--\kapInflWb{}. The two limits lead to two recommendations: record who initiates each contact, and test a fitted node memory against shuffled and model-free benchmarks before
interpreting it. On a high-school record, simulating from the fitted node memory under-predicts
outbreak sizes by up to a factor of $2$, whereas keeping tie reinforcement alone is closer at low transmissibility but overshoots at high.
\end{abstract}

\maketitle

\begin{quotation}
\noindent\textbf{Popular Summary:} How fast an epidemic or rumour spreads depends on \emph{why} people meet repeatedly: because certain individuals are intermittently active (node-driven memory), or because certain pairs repeatedly reconnect (link-driven memory). The two mechanisms can give identical static statistics and still lead to different spreading. We ask the inverse question: given a contact record, can we tell which mechanism is responsible? When the initiator of each contact is observed, the two memories decouple completely. Ordinary proximity data do not record who initiated a contact. There the standard node-memory estimate does not measure memory: it re-describes the distribution of waiting times between contacts. This is not specific to one simple model or to noisy human data. A more flexible model of node memory does not recover the memory either, and the same failure appears in machine-generated networks such as automated web-editing bots and cloud microservices. We measure the problem rather than declare it impossible. Weak traces of timing structure are easy to detect, but they account for less than a tenth of what the standard estimate reports; the rest is the waiting-time distribution. More data does not help, because the error is a bias and not noise: a longer record pins the estimate more precisely on the wrong value. Tie-level memory is the more robust of the two, though not immune. Not knowing who initiated a contact inflates it by a quarter to nearly a half, and the error disappears once initiation is known. Getting the carrier wrong is not harmless: under the conditions we specify it changes predicted outbreak sizes by up to a factor of 2. What these data lack is not direction as such, but information about who acted.
\end{quotation}

\begin{figure*}[t]
  \centering
  \includegraphics[width=\textwidth]{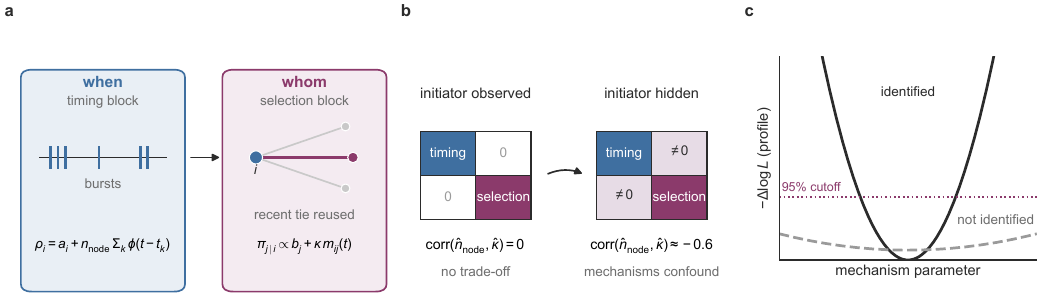}
  \caption{\textbf{Two separable memories.} (a)~Each event is a node firing (timing block, steel blue) followed by a choice of target (selection block, plum): node self-excitation $\nnode$ enters only the timing block, tie reinforcement $\kappa$ only the selection block. (b)~With the initiator observed the Fisher information is block diagonal and $\mathrm{corr}(\nnode,\kappa)=0$ exactly; dropping the label fills the off-diagonal blocks and the two mechanisms confound ($\mathrm{corr}\approx-0.6$ when node memory is present; Sec.~\ref{sec:res_latent}). (c)~The curvature of the profile likelihood, not the location of its optimum, is what decides whether a mechanism is recoverable.}
  \label{fig:theory}
\end{figure*}

\section{Introduction}
\label{sec:intro}

Generative models of temporal networks fall into two broad families. In the first, the
dynamics are carried by \emph{node} properties: each node has an activity potential
governing how often it engages, and links are a by-product of node activation. This is
the activity-driven class~\cite{Perra2012} together with its non-Markovian extensions,
in which a node's propensity to act depends on its own recent
history~\cite{Karsai2014,Ubaldi2016,Ubaldi2017}. In the second, the dynamics are carried
by the \emph{links} themselves: a tie that has been used recently is more likely to be
used again, independently of any node-level state. This is the mechanism behind
edge-level memory models~\cite{Vestergaard2014,Williams2022PRE} and, in a different
tradition, behind dyadic self-exciting point processes and relational event
models~\cite{Butts2008,Hawkes1971}.

The two families are usually treated as modelling choices rather than as competing
hypotheses about data, since both reproduce the coarse statistics that
temporal-network studies typically report: broad activity distributions, bursty
inter-event times, heavy-tailed edge weights, slow topological turnover. Reviews of the
field~\cite{HolmeSaramaki2012,Holme2015,MasudaLambiotte2016} document the proliferation
of mechanisms far better than they document how one would tell them apart.

This paper poses the inverse problem directly. \emph{Given an observed contact sequence, can
the mechanism that carries its memory be identified?} The question matters because the two mechanisms are not interchangeable for spreading: they concentrate transmission differently in time and along different structures, and a control policy that targets bursty individuals is not the one that targets reinforced ties. We find two limits, of different kinds. The first is structural and concerns attribution, i.e., whether node and tie memory can be separated depends on whether the data record which endpoint initiated each contact. The second challenge is empirical and interpretational. The fitted node memory reflects the distribution of inter-event waiting times rather than event ordering; consequently, it cannot be interpreted as self-excitation, even when the underlying mechanisms are theoretically separable. In contrast, tie reinforcement remains recoverable. We identify practical implications for each regime by proposing both specific data collection strategies and diagnostic tests for fitted models. Finally, we show that misidentifying memory on an inappropriate carrier alters simulated outbreak sizes in high-school records by up to a factor of $2$. 

\emph{Framework.} A contact record is a sequence of events $(t_e,i_e,j_e)$. Each
tuple describes node $i_e$, the \emph{initiator}, contacting a target $j_e$ at time $t_e$
[Fig.~\ref{fig:theory}(a)]. Every event has two components. The \emph{timing} component says when the
initiator acts, and it carries node memory. In the self-exciting (Hawkes) form we use, each event elevates its own future rate by an amount summarized by the branching
ratio $\nnode$, representing the expected number of further events triggered by one occurrence. The \emph{selection} component determines the contact target and encodes tie memory. Specifically, recently used ties are chosen more frequently; a target contacted just moments ago is $1+\kappa$ times more likely to be selected than an average target. Additionally, each node exhibits a memoryless baseline activity rate. Both $\nnode$ and $\kappa$ are dimensionless parameters fixed by explicit normalization (Sec.~\ref{sec:model}). 
A common alternative formulation mixes node-driven and link-driven intensities with a weight intended to represent the underlying mechanism. However, this approach fails identification because the mixing weight absorbs arbitrary scales in node activities. In our current parametrisation, the issue becomes concrete: \emph{can we estimate $\nnode$ and $\kappa$ from contact records, and do these estimates correspond to their nominal meanings?} 

\emph{Data.} We analyze twelve datasets varying in three respects: whether initiators are recorded, the burstiness of timing, and the known presence of node memory (Supplementary Table~S1). Four datasets comprise proximity records where initiators are unobserved; these include human face-to-face interactions from the SocioPatterns collaboration in a high school, a scientific conference, and a hospital ward~\cite{Mastrandrea2015,Genois2018,Vanhems2013}, as well as baboon groups~\cite{Gelardi2020}. Four datasets consist of human communication logs where the sender acts as the initiator: institutional email~\cite{Paranjape2017}, online messaging~\cite{Panzarasa2009}, Stack Overflow interactions~\cite{Paranjape2017}, and Wikibooks edits~\cite{Kunegis2013}. Two datasets are algorithmic, i.e., Wikibooks edits by bot accounts and calls among cloud microservices~\cite{Luo2021}. The final two are neural, i.e., a simulated spiking network where presynaptic neurons initiate, and spike trains recorded in mouse visual cortex~\cite{Siegle2021}. Throughout, synthetic records generated from the model with known parameters serve as reference benchmarks.

\emph{Outline.} Section~\ref{sec:results} first tests the separation result on synthetic data,
with and without initiator labels (Secs.~\ref{sec:res_synthetic} and~\ref{sec:res_latent}),
and asks how many events per node the estimator needs (Sec.~\ref{sec:res_sweep}). It then fits the empirical records and tests the fits (Secs.~\ref{sec:res_empirical}--\ref{sec:res_gof}),
develops the diagnostic that shows what the fitted node memory measures (Sec.~\ref{sec:plane}), and quantifies what is at stake for spreading (Sec.~\ref{sec:epidemic}).
Section~\ref{sec:discussion} discusses implications and limitations, and
Sec.~\ref{sec:methods} gives the model, the identifiability results, the estimators and the
data.

\section{Results}
\label{sec:results}

One result structures this analysis. Proposition~\ref{prop:orthogonality} (Sec.~\ref{sec:identifiability}; proof in Supplementary Material, Sec.~S9) establishes that when every event's initiator is observed, the log-likelihood decomposes into a timing term depending solely on node parameters ($\nnode$, its memory timescale, and baseline activities) and a selection term depending solely on tie parameters ($\kappa$ and its timescale). Consequently, the Fisher information matrix becomes block diagonal: estimates of $\nnode$ and $\kappa$ are asymptotically independent, preventing either from absorbing the other. When initiators remain unobserved, the likelihood must average over unknown labels. This averaging couples the two blocks via a missing-information term. The magnitude of this term scales with the entropy of the posterior distribution for the initiator: zero when the initiator is certain and maximal when both endpoints are equally likely. We first test these statements on synthetic data.

Throughout, we employ maximum likelihood estimation with each node's baseline activity profiled out (Sec.~\ref{sec:estimation}); when initiators are unobserved, a Monte Carlo expectation--maximisation algorithm imputes them (Sec.~\ref{sec:em}). A \emph{realisation} denotes a single contact sequence simulated from the model using known parameters, $N$ nodes, and an observation window of length $T$. We test the two \emph{boundary nulls}, defined as $\nnode=0$ (no node memory) and $\kappa=0$ (no tie memory), using likelihood ratio tests (Sec.~\ref{sec:tests}).

\subsection{Does the estimator recover both mechanisms when every initiator is observed?}
\label{sec:res_synthetic}

A realisation with $N=\synN{}$, $T=\synT{}$, $\nnode=\synTrueNnode{}$,
$\kappa=\synTrueKappa{}$ produced $\synEvents{}$ events. The profiled maximum-likelihood
estimates are
\begin{equation}
\begin{split}
  \hat{n}_{\mathrm{node}} &= \synObsNnode \pm \synObsNnodeSE, \\
  \hat\kappa &= \synObsKappa \pm \synObsKappaSE ,
\end{split}
\end{equation}
both within one standard error of the truth, and the across-parameter correlation is
exactly zero, as Proposition~\ref{prop:orthogonality} requires. Both boundary nulls are
rejected.

Figure~\ref{fig:profiles}(c) shows the profile likelihoods for $\nnode$ and $\kappa$,
plotted as the likelihood-ratio statistic $2\Delta\ell$ against each parameter normalised to
its maximum-likelihood estimate, with all remaining parameters re-optimised at each grid
point. Both are sharply curved and cross the $\chi^2_{1,0.95}$ threshold within a narrow
interval that contains the true value, so each mechanism is well identified at this sample
size, and the $\nnode$ profile is invariant to $\kappa$ and conversely.

\subsection{Cost of the missing initiator label}
\label{sec:res_latent}

Removing the labels from the same realisation and refitting by Monte Carlo EM gives
$\hat{n}_{\mathrm{node}} = \synLatNnode$ and $\hat\kappa = \synLatKappa$, with standard errors inflated by
factors $\synSEInflationNnode$ and $\synSEInflationKappa$ respectively. The across-imputation
correlation moves from exactly zero to
\begin{equation}
  \mathrm{corr}(\hat{n}_{\mathrm{node}}, \hat\kappa) = \synLatCorr ,
\end{equation}
a substantial confound: node memory and tie memory become partially exchangeable, with
node self-excitation biased downward and tie reinforcement upward.

This gap combines the information loss from the missing labels with the bias of the
sequential imputation used in the EM algorithm. We separate the two with a
Metropolis-within-Gibbs smoother that samples the exact label posterior (Supplementary
Material, Sec.~S3). On synthetic data with node memory the confound survives the
smoother and tightens to $\mathrm{corr} = \confSynSmooth{}$, so it is real information loss
rather than an artefact of the imputation.

The confound is the off-diagonal block of the missing information. With the initiator
observed the Fisher information is block diagonal and the correlation is exactly zero;
marginalising the latent label $z$ subtracts the missing information (Louis' identity),
$\mathcal I_{\mathrm{obs}}=\mathcal I_{\mathrm{complete}}-\mathrm{Cov}_{z\mid\mathrm{data}}[\nabla_\theta\ell]$,
whose timing--selection block is the entire confound. Its sign follows from \emph{explaining
away}: a quick recurrence of a contact can be credited to node self-excitation or to tie
reinforcement, so an imputation that assigns it to one mechanism debits the other and the
estimates are coupled negatively. The magnitude scales with the entropy of the initiator posterior; we denote its average over events as $\bar H$. Hence, a fraction $\rho$ of true labels validates both predictions (Fig.~\ref{fig:confound}): the spread across imputations collapses proportional to $(1-\rho)\bar H$, while both estimates shift monotonically toward their true values. Thus, the confound is not a fixed scalar but an information-theoretic quantity that vanishes as labels are supplied (Supplementary Material, Sec.~S14)

\begin{figure*}[t]
  \centering
  \includegraphics[width=0.82\textwidth]{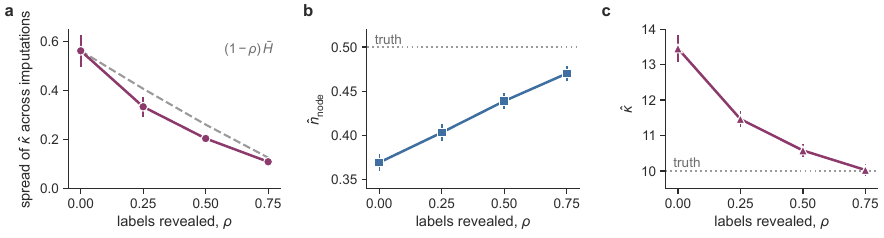}
  \caption{\textbf{The latent-initiator confound is the off-diagonal of the missing
  information, controlled by the initiator-posterior entropy.} A fraction $\rho$ of the true
  initiator labels is revealed and the rest imputed. \textbf{(a)}~The across-imputation spread
  $\sigma_{\hat\kappa}$ (the missing-information magnitude) collapses as labels are revealed,
  tracking $(1-\rho)\bar H$ (dashed), with $\bar H$ the mean per-event entropy of the initiator
  posterior (slope $\approx1$, $r=0.87$). \textbf{(b,c)}~Both estimates de-bias monotonically
  towards the truth (dotted; $\nnode=0.5$, $\kappa=10$): the missing label biases node
  memory down and tie memory up, and revelation removes the bias. Mean $\pm$ s.e.\ over six
  realisations.}
  \label{fig:confound}
\end{figure*}

\subsection{Where the estimator is confidently wrong}
\label{sec:res_sweep}

The previous tests used one sample size. To find where it stops being reliable, we
mapped the coverage of the nominal $95\%$ Wald intervals over $(N, T, \nnode, \kappa)$, with
$\sweepReps{}$ independent realisations per cell. Here $T$ is the length of the observation
window; at fixed activity it sets the number of events each node contributes. Coverage is at
or near nominal in $\sweepCellsOK{}$ of $\sweepCells{}$ cells, and where it degrades the
pattern differs between the two mechanisms [Fig.~\ref{fig:coverage}(a,b)].

Coverage for $\kappa$ remains nominal across all cases, including the shortest observation window [Fig.~\ref{fig:coverage}(b)]. Coverage for $\nnode$ depends on $T$ rather than $\nnode$ itself [Fig.~\ref{fig:coverage}(a)]. At $T=50$, where each node records only ${\sim}2$--$4$ events, coverage collapses to $0.23$--$0.56$ across the full range of $\nnode$. By contrast, at $T=600$ (${\sim}30$ events per node), coverage recovers to the nominal level. The total number of events does not predict failure: cells with coverage below $0.90$ occur even with as few as $\sweepEventsLow{}$ or as many as $\sweepEventsHigh{}$ events, whereas nominal coverage can be achieved with as few as $\sweepEventsMinOK{}$ events. The determining factor is the number of events per node.

The failure manifests as bias rather than underestimated variance. In these degraded cases, the ratio of the across-realisation standard deviation of $\hat{n}{\mathrm{node}}$ to its asymptotic standard error remains near unity (Supplementary Table~S4), whereas $\hat{n}{\mathrm{node}}$ exhibits a downward bias of $0.13$--$0.16$. Consequently, the confidence intervals possess correct widths but are centered away from the true value. This pattern matches the incidental-parameter (Neyman–Scott) signature arising from profiling one baseline activity per node given few events (Lemma~\ref{lem:bracket}). Substituting this profiling with a hierarchical (empirical-Bayes) prior on activities, while leaving all else unchanged, roughly halves the bias and restores coverage at $T=50$ to $0.83$--$0.93$ [dashed curve in Fig.~\ref{fig:coverage}(a); Supplementary Material, Sec.~S3]. The solution is practical: employ the hierarchical prior or ensure several tens of events per node. Since the estimated coverage proportion from $\sweepReps{}$ realisations has a standard error near $0.09$, this claim relies on the gradient across the experimental design rather than any single cell.

\begin{figure*}[t]
  \centering
  \includegraphics[width=\textwidth]{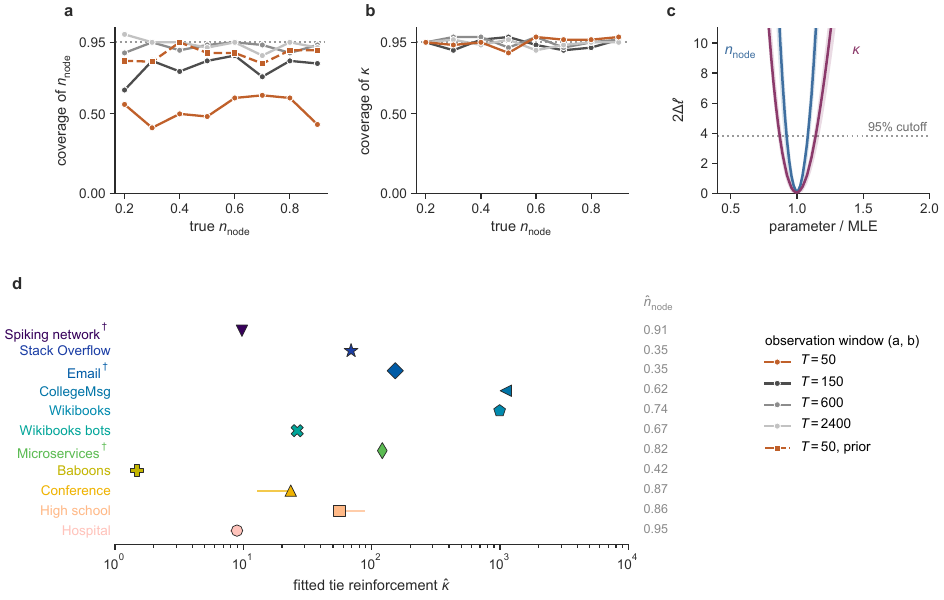}
  \caption{\textbf{Identifiability on synthetic and empirical data.} \emph{(a,b)}~Coverage of the nominal $95\%$ Wald interval for $\nnode$ and $\kappa$ against the true $\nnode$ (each point aggregates the $N$ and $\kappa$ cells; $30$ realisations per cell), for observation windows $T=50$ (orange), $150$, $600$ and $2400$ (darker to brighter), and the dashed $T=50$ curve with a MAP prior. In~(a), coverage of $\nnode$ collapses at $T=50$ (${\sim}2$--$4$ events per node) and is nominal by $T=600$, so the failure is set by events per node, not by $\nnode$; the prior restores coverage (Sec.~\ref{sec:res_sweep}). In~(b), coverage of $\kappa$ stays nominal throughout. \emph{(c)}~Profile likelihoods for $\nnode$ and $\kappa$ on one synthetic realisation, normalised to the maximum-likelihood estimate; both cross the $\chi^2_{1,0.95}$ cutoff within a narrow interval (Sec.~\ref{sec:res_synthetic}). \emph{(d)}~Fitted tie reinforcement $\hat\kappa$ of each record on a logarithmic axis, in the chromatic order of the palette. The point is the day-0 fit of each record; for the three SocioPatterns records the bar is the range of $\hat\kappa$ over observation days with at least $1000$ events. A dagger marks a record with an admissible bound active (indicative only). The fitted $\hat{n}_{\mathrm{node}}$ is printed at the right of each row for completeness and is not interpreted: it reads the inter-event-time distribution rather than node memory (Sec.~\ref{sec:plane}). The cortical recording is omitted because it has spike times but no targets, so no $\kappa$ is fitted. Wikibooks is a bipartite editor--page log (Supplementary Material, Sec.~S3).}
  \label{fig:coverage}\label{fig:profiles}\label{fig:plane}\label{fig:identifiability}
\end{figure*}

\subsection{Empirical estimates}
\label{sec:res_empirical}

We start with the three SocioPatterns records, which share
the same instrument. In the high school, sessionisation at $\sessionGap{}$~s
(Sec.~\ref{sec:data}) reduces the \hsRawSamples{} raw proximity samples to contact events, and
\hsNodesActive{} of \hsNodes{} nodes pass the activity filter. Fitting day~0 alone
($\hsDayZeroContacts{}$ contacts over $24$~h) gives
\begin{align}
  \hat{n}_{\mathrm{node}} &= \hsDayZeroNnode \pm \hsDayZeroNnodeSE, &
  1/\hat\gnode &\approx \hsDayZeroNodeMemory~\mathrm{s}, \\
  \hat\kappa &= \hsDayZeroKappa \pm \hsDayZeroKappaSE, &
  1/\hat\gedge &\approx \hsDayZeroTieMemory~\mathrm{s} .
\end{align}
Taken at face value, node activity would be strongly self-exciting with a memory of a few
minutes, and target choice strongly reinforced over several hours, a just-contacted target
being roughly $1+\hat\kappa$ times as attractive as an average one. We report these as fitted
values only. Sections~\ref{sec:res_gof} and~\ref{sec:plane} show that the timing part of the
model is rejected on these data and that $\hat{n}_{\mathrm{node}}$ measures the inter-event-time distribution rather than memory; $\hat\kappa$ survives these tests where the selection diagnostics pass.

Across the five observation days the high school gives $\hat{n}_{\mathrm{node}} = \perdayHsNnode{}$ and
$\hat\kappa = \perdayHsKappa{}$. Almost all of the spread in $\hat{n}_{\mathrm{node}}$ comes from day~4, a
half day with $5.0$~h of observation, for which $\hat{n}_{\mathrm{node}}\approx\schoolShortDayNnode{}$;
days 0--3 give $0.86$--$0.87$. The lower value on the short day is consistent with the
downward finite-sample bias of Sec.~\ref{sec:res_sweep}, since each node then contributes
fewer events. The same per-day protocol gives $\hat{n}_{\mathrm{node}} = \perdaySfNnode{}$ and
$\hat\kappa = \perdaySfKappa{}$ for the conference and $\hat{n}_{\mathrm{node}} = \perdayHoNnode{}$ and
$\hat\kappa = \perdayHoKappa{}$ for the hospital ward.

The latent-initiator confound described in Sec.~\ref{sec:res_latent} is significantly weaker on these records. For high-school day~0, the across-imputation correlation is $\hsDayZeroCorr{}$, compared to $\synLatCorr{}$ for synthetic data of comparable structure. This reduced magnitude reflects the node memory inherent in the data; where inter-event times dominate the signal (Sec.~\ref{sec:plane}; Supplementary Material, Sec.~S3), estimates of $\hat{n}_{\mathrm{node}}$ exhibit minimal sensitivity to imputed labels.

Figure~\ref{fig:plane}(d) plots the fitted tie reinforcement ($\hat\kappa$) for all records, using either latent-initiators or observed-initiators (Sec.~\ref{sec:estimation}). Values span nearly three orders of magnitude: $\hat\kappa \approx 1.5$ (baboons), $\approx 9$ (hospital), $23$ [range $13$--$23$ over two days] (conference), $56$ [range $56$--$88$ over its duration] (high school), $26$ and $121$ (Wikibots bots/microservices), and $69$--$1111$ (communication records: Stack Overflow $69$, email $153$, Wikibooks $994$, CollegeMsg $1111$). In the simulated spiking network, $\hat\kappa \approx 10$. For email, microservice, and spiking-network records, an admissible bound is active (marked $\dagger$ in Fig.~\ref{fig:plane}); their reported $\hat{n}_{\mathrm{node}}$ or $\hat\kappa$ are thus indicative only. The cortical recording lacks targets and is not fitted for $\kappa$.

\subsection{Sensitivity to preprocessing}
\label{sec:res_sensitivity}

Three design choices influence estimates (Sec.~\ref{sec:data}). Without sessionisation, continuous conversations register as individual contacts every $\samplingGrid{}$~s; the model then fits the instrumentation rather than behavior: the tie-memory rate locks to the grid ($\hat\gedge \approx 0.05\,\mathrm{s}^{-1}$) and $\hat\kappa$ diverges. Merging samples from the same pair within $\Delta_{\mathrm{sess}} = \sessionGap{}$~s corrects this, yielding a tie-memory time of roughly $7.8$~h for high-school day~0. Varying $\Delta_{\mathrm{sess}}$ between $40$ and $120$~s alters event counts by about $15\%$ while keeping $(\hat{n}_{\mathrm{node}},\hat\kappa)$ within their between-day spread. Filtering out nodes with fewer than $\minContacts{}$ contacts impacts $\hat\kappa$: such nodes contribute negligible timing likelihood but expand the alternative pool in the selection normaliser $Z_i$ (Eq.~\eqref{eq:pi}), lowering $\hat\kappa$. Fitting data per day prevents stationary kernels from spanning overnight gaps, avoiding a bias of both memory rates toward the daily cycle.

\subsection{Does the model describe the data well enough for its parameters to mean anything?}
\label{sec:res_gof}

A fitted parameter is a mechanism estimate only if it does. Under the
time-rescaling theorem, the compensated inter-event times of a correctly specified conditional
intensity are independent unit exponentials. We test their distribution by a
Kolmogorov--Smirnov (KS) test and their independence by the within-node lag-one correlation
(reported as a $z$-score, $z_{\mathrm{acf}}$), and we test the selection block through the
randomised probability integral transform of the realised target under the fitted
$\pi_{j\mid i}$.

Both tests are calibrated on synthetic data. A correctly specified fit passes both, and a circadian baseline fitted with constant activities is caught decisively. A moderate power-law node kernel fitted with an exponential is missed by both ($p = \gofPowerLawKS{}$,
$z = \gofPowerLawACF{}$ on $\gofPowerLawEvents{}$ events). A pass therefore means only that the model is not grossly misspecified in the ways these tests see; it says nothing about the shape of the kernel.

Fitting the first $\gofHsWinHours{}$~h window of high-school day~0 ($\gofHsWinContacts{}$ contacts, $N=\gofHsWinN{}$), so that no overnight gap is
involved, gives
\begin{equation}
\begin{split}
  \mathrm{KS} &= \gofKSHighSchool, \quad
  z_{\mathrm{acf}} = \gofACFHighSchool, \\
  \overline{\tau} &= \gofMeanTau, \quad \mathrm{Var}(\tau) = \gofVarTau ,
\end{split}
\end{equation}
against $\overline{\tau} = \mathrm{Var}(\tau) = 1$ under a correct model. Both decay rates
are interior to the bounds of Eq.~\eqref{eq:gammabounds}, so this is not a boundary artefact.
The failure has two separate causes.

\emph{(i) Simultaneous events.} A fraction $\gofZeroAtom{}$ of the rescaled times is exactly
zero, and this atom accounts for the entire KS statistic (Supplementary Material, Sec.~S2). A
group conversation is recorded as several simultaneous dyads, so $\gofTiedFraction*{}$ of
events share a time and an initiator with an earlier event and the initiator has inter-event
times of exactly zero. This is a property of the data encoding, not of the memory mechanism.

\emph{(ii) Temporal structure.} Removing the tied residuals reduces the KS statistic to
$\gofKSHighSchoolDetied{}$ but leaves the serial correlation at $z = \gofACFHighSchoolDetied{}$.
The negative sign means that the fitted self-excitation is too strong at short lags: after a
burst the model expects the next event sooner than it arrives. Treating a co-present group as
one node activation with several targets removes the atom (zero-atom fraction
$\groupAtomDy{}\!\to\!\groupAtomGp{}$) but not this defect: the serial-correlation $z$-score
stays large ($|z|:\groupZdy{}\!\to\!\groupZgp{}$, in a $40$-node active subnetwork) and $\hat{n}_{\mathrm{node}}\approx\groupNnode{}$ is
unchanged (Supplementary Material, Sec.~S3). The residual is a separate misspecification of the
timing kernel.

\begin{figure*}[t]
  \centering
  \includegraphics[width=\textwidth]{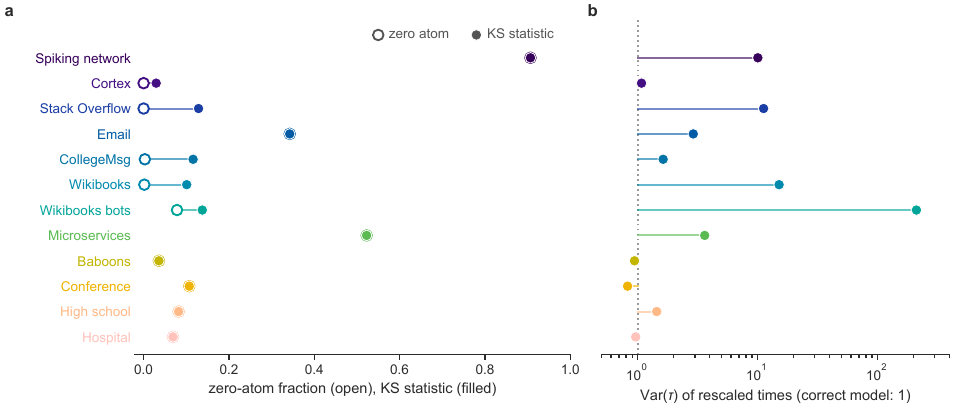}
  \caption{\textbf{The two causes of misfit across datasets.} One row per record, in the chromatic order of the palette (darkest at the top). \emph{(a)}~Zero-atom fraction of the rescaled times (open markers) and Kolmogorov--Smirnov statistic (filled). Where the two coincide the misfit of the marginal \emph{is} the atom (the SocioPatterns records, baboons, email, microservices and the simulated spiking network); where the filled marker lies to the right of the open one there is a residual \emph{shape} defect (CollegeMsg, Stack Overflow, Wikibooks, Wikibooks bots). The cortical spike trains fit well ($\mathrm{KS}=0.03$). \emph{(b)}~$\mathrm{Var}(\tau)$, equal to one under a correct model (dotted). The three SocioPatterns records, the baboons and the cortex lie between $0.8$ and $1.5$; CollegeMsg ($1.6$), the email log ($2.9$), the microservices ($3.6$), the simulated network ($10$), Stack Overflow ($11$), Wikibooks ($15$) and the Wikibooks bots ($2\times10^{2}$) are over-dispersed.}
  \label{fig:identity}
\end{figure*}

Figure~\ref{fig:identity} repeats the decomposition across datasets. The identity between the
KS statistic and the zero-atom fraction holds in all three proximity records, in the email log
and in the spiking network; a dataset in which the two came apart would falsify the two-cause
reading. The online records (Stack Overflow, Wikibooks) have no atom but a large shape defect
and strong over-dispersion [$\mathrm{Var}(\tau) \gg 1$], i.e.\ heavy-tailed temporal clustering
that the exponential kernel does not capture. Parametric-bootstrap calibration
($\gofBootReps{}$ replicates) gives the following verdicts. The timing block is rejected in
every proximity record (calibrated $p=\gofHsCalibP{}$ for the high school, $\gofHoCalibPTiming{}$
for the hospital ward and $\gofSfCalibPTiming{}$ for the conference), with KS statistics three to
five times the null medians ($\gofHsNullMedianKS{}$, $\gofHoNullMedianKS{}$,
$\gofSfNullMedianKS{}$). The selection block is rejected in the high school
($p=\gofHsCalibPSel{}$) but passes in the hospital ward ($p=\gofHoCalibPSelection*{}$) and the
conference ($p=\gofSfCalibPSelection*{}$).

\paragraph*{A repeat test for the selection block.} Because $\kappa$ is meant to capture
repeated contact, we add a test aimed directly at repeats, shown here on the messaging record,
where $77\%$ of messages repeat a tie (Supplementary Material, Sec.~S2). Two surrogates show
that these repeats are tie memory: a target permutation that keeps every user's in- and
out-strength lowers the repeat fraction to $0.25$, and permuting timestamps against the sequence
of pairs lowers the share of repeats that follow within $10$~min from $0.44$ to $0.006$. The
fitted selection model ($\hat\kappa\approx1.1\times10^3$, tie-memory time ${\approx}16$~h)
reproduces same-direction repeats ($0.64$ of messages simulated against $0.66$ observed) but not
replies ($0.04$ against $0.11$) or the spread of repeat lags (90th percentile $0.9$ against $2.5$
days). The single-timescale tie kernel therefore misses reciprocity and the separation between
conversational and long-term recurrence, which locates the next refinement of the selection
block.

\paragraph*{What the fitted parameters measure.} These results define how to interpret the empirical estimates. We distinguish the following failure modes that prevent fitted parameters from capturing underlying mechanisms: (F1) Missing initiators couple the model blocks (Sec.~\ref{sec:res_latent}); this is the sole structural issue, resolved by recording initiators. (F2) Sparse events per node bias $\hat{n}_{\mathrm{node}}$ (Sec.~\ref{sec:res_sweep}), which a hierarchical prior largely corrects. (F3) Parameters for rejected blocks do not represent mechanisms, regardless of standard errors; here the timing block fails in every record and the selection block in the high school dataset. (F4) Heavy-tailed inter-event times lacking temporal order can still yield large fitted branching ratios (discussed in Sec.~\ref{sec:plane}). In these datasets, $\hat{n}_{\mathrm{node}}$ reflects both F3 and F4; we therefore use it only as a diagnostic coordinate, not as an estimate of self-excitation.

Tie reinforcement is less susceptible to these failures. We report $\hat\kappa$ as a mechanism parameter only when selection diagnostics pass, noting one exception that explains the bias in Fig.~\ref{fig:confound}(c): $\kappa$ remains vulnerable to F1 (missing initiators). While the pair history itself is independent of identity, the likelihood depends on it because choice probabilities are normalised over the initiator's alternatives (Corollary~\ref{cor:dichotomy}). Consequently, a latent initiator biases $\hat\kappa$ upward. For the two directed records where we can hide and then reveal true labels, this inflation equals $\kapInflSo{}$ (Stack Overflow) and $\kapInflWb{}$ (Wikibooks) relative to the revealed-label fit; revealing labels removes it (Supplementary Material, Sec.~S14). In proximity records where labels cannot be revealed, $\hat\kappa$ serves as an upper bound, likely overestimated by tens of percent.

\subsection{What the fitted node memory measures}
\label{sec:plane}
The goodness-of-fit tests say that the timing model is wrong; they do not say what $\hat{n}_{\mathrm{node}}$
is measuring instead. A node's event train has two properties that are easy to conflate. One is
how unevenly its events are spread in time: some nodes act in tight bursts separated by long
quiet gaps. The other is whether the timing carries memory: whether a short gap tends to be
followed by another short gap. These are separate properties, and a model meant to capture the
second can be fooled by the first. An exponential self-exciting model fitted to a node that is
bursty but has no memory still reports strong self-excitation, because a large branching ratio
is its only way to reproduce the uneven spacing.

Two statistics separate the two properties. The burstiness
$B = (\sigma_\tau - \mu_\tau)/(\sigma_\tau + \mu_\tau)$ measures the spread of a node's
inter-event times $\tau$, and the memory coefficient $M$ is the correlation between consecutive
inter-event times~\cite{GohBarabasi2008}, computed on the within-node ranks of the intervals and pooled over nodes (Supplementary Material, Sec.~S15). Ranking removes between-node differences in scale and the dominance of a few very long intervals, which the Pearson correlation of raw intervals does not. A Poisson process has $(B,M)=(0,0)$; a renewal process
with any inter-event law has $M=0$; self-excitation, in which short intervals beget short
intervals, has $M>0$. The test of whether $\hat{n}_{\mathrm{node}}$ reads memory is to refit it after
shuffling the order of each node's inter-event times, which keeps their distribution exactly and
destroys their order. The shuffle ratio
\begin{equation}
  R = \hat{n}_{\mathrm{node}}^{\,\mathrm{shuffled}} / \hat{n}_{\mathrm{node}}^{\,\mathrm{real}}
\end{equation}
is $1$ when $\hat{n}_{\mathrm{node}}$ depends only on the inter-event-time distribution, and $1-R$ is the share
of $\hat{n}_{\mathrm{node}}$ carried by temporal order. The two tests answer different questions. $R$ says
whether $\hat{n}_{\mathrm{node}}$ carries information beyond the inter-event-time distribution, i.e.\ whether
the fitted value would differ for a memoryless renewal process with the same waiting times. $M$
says whether memory is present at all. A value $R\approx1$ therefore shows that $\hat{n}_{\mathrm{node}}$ is
not a measurement of memory; it does not by itself show that memory is absent.

\paragraph*{Calibration.} We calibrated $R$ on single-node event trains with directly
controlled memory: an autoregressive process on the log inter-event time, whose lag-one
coefficient sets $M$ [Fig.~\ref{fig:plane_validation}(a)]. $R$ is close to one at $M\approx0$ and
declines as $M$ grows, to $\fssRatEightLo{}$--$\fssRatEightHi{}$ at $M=0.8$. For exponential self-excitation the
order-carried share depends on how the memory time compares with the typical waiting time. It is
$1-R\approx0.51$ for a single node whose memory time equals its baseline waiting time
($\nnode=0.5$; Supplementary Material, Sec.~S12), but only $0.03$--$0.13$ when exponential memory
is injected into the real high-school activity pattern, where the memory time is several times
shorter than the typical waiting time (positive control below). The decline is smooth and does not sharpen with the number of events per node $E$:
the drop in $R$ across $M:0.1\!\to\!0.6$ is $\fssDropSmallE{}$ at
$E\approx\fssElo{}$ and $\fssDropMidLo{}$--$\fssDropMidHi{}$ for $E\gtrsim600$, and the width of the
decline does not shrink with $E$ [log--log slope $\fssWidthSlope{}$, 95\% interval $\fssWidthSlopeLo{}$ to $\fssWidthSlopeHi{}$, over $E\approx\fssElo{}$--$\fssEhi{}$;
Fig.~\ref{fig:width}(c)]. Over the range of $E$ that real records attain, there is therefore no
threshold in $M$ that would tell a practitioner whether a record resolves node memory; we do not
claim anything about the limit $E\to\infty$. Two features limit what $R$ can show. At the memory
levels found in real data, $M\lesssim0.28$, the calibration curve itself only falls to
$R\approx\fssRatRecLo{}$--$\fssRatRecHi{}$, so the test is weak there. And $R$ responds to memory of the clustering
kind the exponential kernel represents, not to every form of temporal dependence. We therefore
read every record through $R$ and the model-free $M$ together.

The choice of null matters. Reassigning each node's events to random times on the global
timeline also collapses $\hat{n}_{\mathrm{node}}$, but it does so even for a memoryless renewal process
(median $R=\pairShufNull{}$, $5$--$95\%$ range $\pairShufNullLo{}$--$\pairShufNullHi{}$ over $\nullReps{}$ replicates, because the fit to the resulting Poisson-like train is unstable; against $\isiShufNull{}$ for the order shuffle), because it destroys each node's rate structure rather
than its memory. Only the inter-event-order shuffle isolates order.

\paragraph*{A phase space for exponential self-excitation.} For a Hawkes process with baseline
rate $\mu$ and kernel $\gamma e^{-\gamma\tau}$, measuring time in units of $1/\mu$ maps
$(\mu,\nnode,\gamma)$ to $(1,\nnode,g)$ with $g=\gamma/\mu$. Every such process, and
therefore $R$ for it, is a function of $\nnode$ and $g$ alone. This is exact
(time-rescaling invariance), not a modelling choice. A sweep over $\nnode\in[0.1,0.9]$ and
$g\in[0.03,30]$ ($4$ processes of $3{\times}10^{3}$ events per cell; Supplementary Material,
Sec.~S7) shows that $R$ in fact depends on $g$ alone. Every cell collapses onto
$R(g)=R_0+(1-R_0)/[1+(g_0/g)^{a}]$ with $R_0=0.12$, $g_0=1.14$, $a=1.54$ ($\chi^2=55$ for $56$
degrees of freedom), and letting $g_0$ vary with $\nnode$ does not improve the fit
[Fig.~\ref{fig:phase}(a)]. This closed form is a fit, not a derivation. Its scale has a
simple reading: $g_0\approx1$ is where the memory time $1/\gamma$ equals the mean gap $1/\mu$
between independent (immigrant) events. Memory that outlasts that gap couples successive
clusters, and its order is lost under the shuffle ($R\approx0.1$--$0.5$). Memory shorter than the
gap is confined to single clusters, whose structure the inter-event distribution already
encodes ($R\to1$). Hence $R=0.5$ at $g\approx1.0$ and $R=0.9$ at $g\approx4.3$, for every
$\nnode$ [Fig.~\ref{fig:phase}(b)]. The change is a smooth crossover, consistent with
Fig.~\ref{fig:plane_validation}(c), with no sharp boundary.

\paragraph*{Where the records sit on the map.} The fitted position of each record, $g=\hat\gamma/\hat a_i$ (median over nodes) under the latent-initiator fit, places the three SocioPatterns records at $g\approx31$--$72$ and the baboons at $g\approx2.7$. The communication, algorithmic and simulated-spiking records lie at $g\gtrsim200$, where the curve predicts $R\approx1$. For the email, microservice and spiking-network records a timescale sits at an admissible bound, so their $g$ is a lower bound or not identified. Above $g\approx30$ the shuffle ratio is close to one \emph{even for exponential self-excitation}: the positive control, at $g\approx6$--$24$, gives $R=0.87$--$0.97$, at or somewhat below the curve's $0.94$--$0.99$. The shuffle test therefore has little power to detect memory in this range, and for these records the model-free $M$ carries the argument.

\paragraph*{Controls matched to each record.} The curve is calibrated on one homogeneous node in continuous time, so for a record it is a guide and not a prediction. The position above also comes from a different estimator (the latent-initiator fit) than the one that produced $R$ (a single-exponential fit over a grid of rates). We therefore run each record's screen against surrogates (Supplementary Table~S13): exponential self-excitation at the screen's own fitted $(\hat n,\hat\gamma)$ with the record's per-node counts, duration and time resolution; a Poisson control; and a shuffle in blocks of $20$ intervals that keeps slow rate changes. This resolves the two records that sat off the curve. The hospital has $R=0.92$ against the curve's $\approx0.996$ at $g\approx38$. The screen's fit is $\hat n=0.99$ at the slowest rate it admits, the matched surrogate gives $R=0.92$--$0.99$, the Poisson control $1.01$ and the block shuffle $0.99$. The deficit is therefore the share of $\hat n$ carried by slow drift, not a failure of the map. The baboons have $R=1.01$ against $0.72$--$0.89$ at $g\approx2.7$. The screen's fit is $\hat n=0.07$ with the rate at the fastest value of the grid, far below the $20$~s resolution, and the surrogate and the Poisson control both give $R=1.00$. The baboon record does carry order memory (within-node rank coefficient $0.16$, of which $0.04$ survives dividing out the shared activity profile). This estimator does not resolve it: the memory coefficient detects it and $R$ does not. The cortical spike trains ($g\approx1.1$, interquartile range $0.9$--$2.5$) fall in the transition. The curve predicts $R=0.49$--$0.79$ over that range, and the full-shuffle $R=0.73$ lies inside it. The local shuffle, which preserves slow rate changes, gives $R=0.97$, consistent with the slow drift found above. The map locates where each test can and cannot see node memory.

\begin{figure*}[t]
  \centering
  \includegraphics[width=\textwidth]{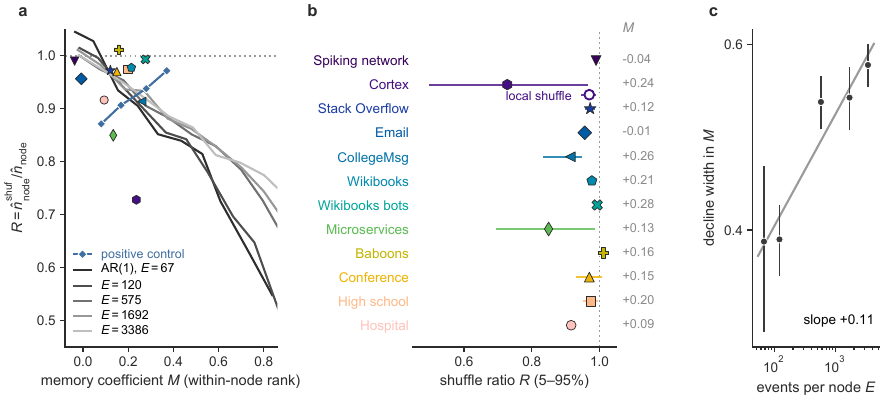}
  \caption{\textbf{Calibration of the shuffle ratio and the position of real data.} \emph{(a)}~Shuffle ratio $R=\hat{n}_{\mathrm{node}}^{\mathrm{shuf}}/\hat{n}_{\mathrm{node}}$ against the memory coefficient $M$ (within-node rank) for an autoregressive generator with injected memory, at increasing event counts per node $E$ (darker to brighter), for exponential node memory injected into the high-school activity pattern (blue diamonds, positive control) and for every record at its measured $(M,R)$ in the chromatic order of the palette. \emph{(b)}~The same $R$ per record with its $5$--$95\%$ interval, and $M$ at the end of each row; for the cortex, filled is the full shuffle and open the local shuffle within blocks of $20$ intervals. The positive control also sits near $R=1$ at the memory levels of real data, so $R$ is read together with $M$. \emph{(c)}~Width of the decline (the $M$ interval over which $R$ falls from $0.9$ to $0.7$) against $E$ on log--log axes, slope $\fssWidthSlope{}$ [$\fssWidthSlopeLo{}$, $\fssWidthSlopeHi{}$] (error bars from a parametric bootstrap): a width that shrank with $E$ would indicate a sharp threshold; the positive slope indicates a smooth crossover.}
  \label{fig:plane_validation}\label{fig:width}
\end{figure*}

\begin{figure*}[t]
  \centering
  \includegraphics[width=0.98\textwidth]{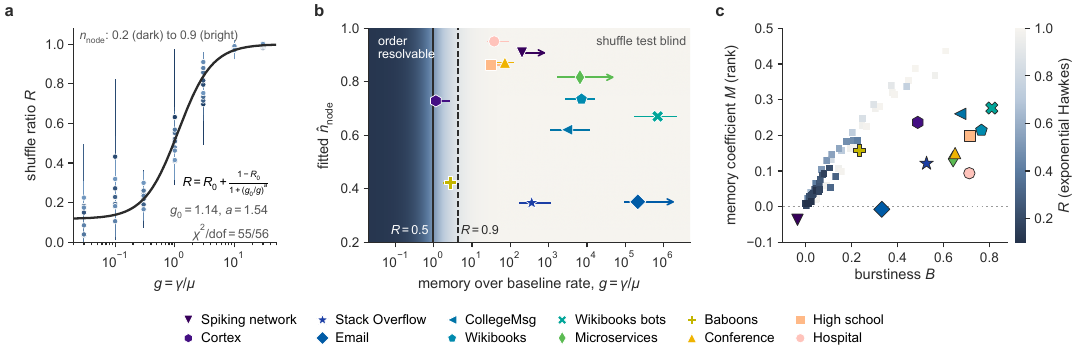}
  \caption{\textbf{A phase space for exponential node memory, and where the records sit.}
  \emph{(a)}~Shuffle ratio $R$ of exponential Hawkes processes against $g=\gamma/\mu$, the memory
  rate over the baseline rate, for $\nnode=0.2$--$0.9$ (light to dark; bars are the spread over four
  processes). All cells fall on one curve (black), the fitted closed form given in the panel.
  \emph{(b)}~The resulting map in the $(g,\nnode)$ plane, from the closed form (shading; solid and
  dashed lines are $R=0.5$ and $0.9$), with each record at its fitted $(g,\hat{n}_{\mathrm{node}})$ and the
  inter-quartile range of $g$ over nodes; an arrow marks records whose $\hat\gamma$ is at an admissible bound, so that $g$ is a lower bound. \emph{(c)}~The Hawkes cells (shaded by $R$) and the
  records on the burstiness--memory plane ($M$: within-node rank).}
  \label{fig:phase}
\end{figure*}

\paragraph*{Proximity records.} The three SocioPatterns records are bursty ($B\approx0.7$) and carry weak order memory:
$M=\planeMhospital{}$ (hospital), $\planeMconf{}$ (conference) and $\planeMschool{}$ (school). It is zero after the shuffle and positive at lags
up to $2$, $4$ and $5$ (Supplementary Tables~S12 and~S7). Dividing out the activity profile shared by all nodes does not remove it
($+0.08$ to $+0.22$), and neither does collapsing repeated contacts with the same partner ($+0.07$ to $+0.16$). The positive control below sets the scale:
exponential self-excitation with branching ratio $0.2$--$0.45$ gives $M=0.08$--$0.17$, and $0.37$ at $\nnode=0.8$.

\paragraph*{The shuffle on proximity records.} Shuffling each node's inter-event order changes $\hat{n}_{\mathrm{node}}$ by little: $R=0.92$ [$0.91$, $0.93$] for the
hospital, $0.98$ [$0.96$, $0.99$] for the school, $0.97$ [$0.94$, $1.00$] for the conference and $1.01$ for the baboons. The school value is stable across
subsamples ($R=\rawHsWinR{}$; $5$--$95\%$ over twelve windows $\rawHsWinLo{}$--$\rawHsWinHi{}$). Part of this is the simultaneity atom, a feature of the marginal.
After collapsing each node's coincident contacts into one activation, with $M$ unchanged, $R$ falls to $\deTieHoR{}$ (hospital), $\deTieSfR{}$ (conference, windows
$\deTieSfLo{}$--$\deTieSfHi{}$) and $\deTieHsR{}$ (school, windows $\deTieHsLo{}$--$\deTieHsHi{}$). With only some tens of events per node, and with $R$ close to one
even for memory (positive control below), the shuffle test alone cannot decide whether memory is present. That rests on the memory coefficients, which are
model-free and need no initiator label.

\paragraph*{What $\hat{n}_{\mathrm{node}}$ reads on these records.} The fitted $\hat{n}_{\mathrm{node}}\approx0.9$ is far larger than the memory the order carries. It is largely
the exponential kernel reproducing a heavy-tailed inter-event distribution with a zero atom, and the weak order that is present is not what it measures. The
latent-initiator confound of Sec.~\ref{sec:res_latent} could in principle couple the two axes here, since Proposition~\ref{prop:orthogonality} guarantees
orthogonality only when initiators are observed. It is weaker on these records (correlation $\hsDayZeroCorr{}$ against $\synLatCorr{}$ on synthetic data), as expected if
$\hat{n}_{\mathrm{node}}$ responds little to the imputed labels, but we do not claim it is absent. The \emph{selection} channel carries strong memory: in the high school
$\tieReuseSchool{}$ of contacts repeat an earlier tie (each tie is used ${\sim}\tieUsesPerEdge{}$ times on average), and $\hat\kappa$ is large in all three records
(Sec.~\ref{sec:res_empirical}).

\paragraph*{Directed and algorithmic records.} The communication logs do carry model-free
memory: Stack Overflow has $M=+0.12$, CollegeMsg $+0.26$, Wikibooks $+0.21$ and the Wikibooks bots $+0.28$, each above its shuffle null at several lags (Supplementary Material, Sec.~S3). The email log has none ($-0.01$), and the microservice call graph is weak and uncertain ($+0.13$, $5$--$95\%$ over nodes $+0.01$ to $+0.21$). The fitted branching ratio responds to order here, but weakly. Over subsamples and shuffle seeds ($5$--$95\%$ intervals), the
exponential kernel gives $R=0.96$ [$0.94$, $0.96$] for the email log, $0.91$ [$0.84$, $0.95$] for
CollegeMsg, $0.97$ [$0.96$, $0.98$] for Stack Overflow, $0.98$ [$0.97$, $0.98$] for Wikibooks and
$0.99$ [$0.99$, $1.00$] for the Wikibooks bots. The microservice call graph, $R=0.85$ [$0.70$,
$0.98$], has the widest interval (Supplementary Material, Sec.~S12). Most of these intervals
exclude one, so temporal order is detected, but it accounts for less than $10\%$ of $\hat{n}_{\mathrm{node}}$
in every human record. In most of these records memory is present ($M>0$), yet the fitted branching ratio
is still set almost entirely by the inter-event-time distribution.

\paragraph*{Neural records.} The simulated spiking network carries no node order by this measure ($M=\snnM{}$), and $\hat{n}_{\mathrm{node}}$ barely responds to the shuffle: $R=0.99$ $[0.99,0.99]$. Re-tuned into an asynchronous, refractory regime it has $M=\snnRefM{}$ ($z_{\mathrm{acf}}=\snnRefAcfZ{}$) and again $R=0.99$ $[0.99,0.99]$. The cortical spike trains
are the one record in which the full shuffle moves $\hat{n}_{\mathrm{node}}$ substantially. Over six random
sets of ten well-isolated units, each with all of its spikes in a five-minute spontaneous block
(${\sim}2000$ spikes per unit), $M=+0.24$ and $R=0.73$ $[0.50,0.96]$ ($0.68$ with the scale-free
kernel). The order involved is slow, however. Shuffling intervals only within blocks of $20$
consecutive intervals, which keeps each unit's slow changes in firing rate, gives $R=0.97$
$[0.95,0.99]$: the large share comes from rate fluctuations over seconds that a stationary model
absorbs into $\hat{n}_{\mathrm{node}}$, as it absorbs the daily cycle in the contact records
(Sec.~\ref{sec:res_sensitivity}), and short-range order carries about $3\%$, as elsewhere
(Supplementary Material, Sec.~S7).

\paragraph*{A flexible kernel, and the size of the effect.} A flexible scale-free
(sum-of-exponentials) node kernel fits the timing better but does not change the conclusion
(Supplementary Material, Sec.~S12). The largest order-carried share it finds is $17\%$, for
Stack Overflow [$R=0.83$, $(0.74,0.91)$], and on a synthetic power-law self-exciting process
its shuffle ratio [$0.91$, $(0.83,1.00)$] overlaps that of a memoryless renewal control with the
same inter-event distribution [$0.98$, $(0.94,1.04)$], so it cannot certify scale-free node
memory even when it is present. The limit is an effect size, not a sample size. Measuring the
order statistic $\Delta=\hat{n}_{\mathrm{node}}-\hat{n}_{\mathrm{node}}^{\mathrm{sh}}$ against a matched renewal null,
temporal order is easy to detect: $\Delta\approx0.011$ at a true $\nnode=0.5$, and
${\sim}10^{3}$ events reject $\Delta=0$ with $80\%$ power. But for a scale-free kernel $\Delta$
is only $2$--$3\%$ of $\hat{n}_{\mathrm{node}}$, and it does not grow with the number of events
(Supplementary Material, Sec.~S12). The contamination by the inter-event distribution is
therefore a bias, not a variance: a longer record sharpens $\hat{n}_{\mathrm{node}}$ around a value set by
that distribution rather than moving it towards the memory it is taken to measure.

\paragraph*{Positive control.} When node memory is present and of the assumed form, the
estimator recovers it. Taking the per-node activity of the high-school record
(${\sim}\posConEvNode{}$ events per node) and injecting a known exponential node memory (memory
time $30$~s), and separately a known tie reinforcement, the estimator recovers both across their
range: $\nnode$ from $\posConNlo\!\to\!\posConNloFit$ to $\posConNhi\!\to\!\posConNhiFit$, and
$\kappa$ from $\posConKlo\!\to\!\posConKloFit$ to $\posConKhi\!\to\!\posConKhiFit$
(Supplementary Material, Sec.~S5). The shuffle ratio of the injected node memory is
$R=0.87$, $0.91$, $0.94$ and $0.97$ at $\nnode=0.2$, $0.4$, $0.6$ and $0.8$ (three shuffle seeds,
spread ${\le}0.01$), i.e.\ close to one and rising with $\nnode$, as for the exponential Hawkes
processes of the Supplementary Material, Sec.~S12: stronger self-excitation makes the
inter-event distribution burstier, and the shuffle preserves that distribution. Even genuine,
correctly specified node memory is thus carried mostly by the inter-event-time distribution in
this setting, which is why the positive control separates from the proximity records through
$M$ and not through $R$.

\paragraph*{Summary.} Both axes of the model have a resolvable and an unresolvable region,
separated by smooth crossovers rather than thresholds. On the node axis, $\hat{n}_{\mathrm{node}}$ is recovered
correctly when the model is right, but on its own it cannot distinguish self-excitation from a
renewal process with the same waiting times, so the presence of node memory has to be
established by a model-free statistic such as $M$; on the
tie axis, the information about $\kappa$ per event is intensive and smooth in the number of
available partners, so $\kappa$ becomes identifiable with enough events (Supplementary Material,
Sec.~S6). The empirical records sit on the unresolvable side of the node axis and the resolvable
side of the tie axis. This asymmetry, node memory unresolved and tie memory recovered, is what these data can and cannot tell us about mechanism.

\subsection{Why the distinction matters: spreading}
\label{sec:epidemic}

The introduction asserts that node memory and tie memory are not interchangeable.
We test this in three ways: on synthetic networks that differ only in where the memory sits, on
surrogates that hold the static network fixed, and on a real record.

\paragraph*{Pure poles.} We generate two temporal networks matched on the coarse statistics:
the same $N$, the same baseline activities (hence the same aggregate event rate) and a uniform
base attractiveness, so that the only difference is where the memory sits. The \emph{node}
pole ($\nnode>0$, $\kappa=0$) has bursty node timing and uniformly random targets; the
\emph{link} pole ($\nnode=0$, $\kappa>0$) has memoryless Poisson timing and reinforced targets.
Spreading differs systematically between them (Fig.~\ref{fig:epidemic}). In this setting
(first row of Table~\ref{tab:regimes}), at $\beta=\epiBeta{}$ SIR infects a fraction
$\epiNode{}$ of the population on the node pole and $\epiLink{}$ on the link pole, and the SIS
endemic prevalence is $\epiSisNode{}$ versus $\epiSisLink{}$ (the node pole is marginally larger in SIR below $\beta\approx0.1$, where final sizes are small, and the link pole is larger above). The difference is larger for
complex contagions, which we include to test whether nonlinear transmission amplifies it. In a
simplicial contagion, transmission is amplified by infected mutual contacts; in a threshold
contagion, repeated exposures must accumulate before they decay. Both depend on repeated
interaction within a short time, which the link pole supplies and the node pole, whose bursts
are scattered across random targets, does not; the link pole accordingly spreads both
faster at moderate and high transmissibility [Fig.~\ref{fig:epidemic}(c,d)].

\paragraph*{Matched surrogates.} Moving memory between the timing and the selection channel also
changes aggregate edge weights, repeated-contact structure and strength heterogeneity, so a
spreading difference might be a by-product of the different static networks the two mechanisms
produce. We separate the two with surrogates of each realisation that hold named statistics
exactly fixed (Supplementary Material, Sec.~S10). Permuting timestamps across events keeps the
static graph, every edge weight and every node strength, and destroys all temporal structure.
Under this surrogate the node--link gap in SIR final size at $\beta=0.35$ falls from
$0.069\pm0.029$ ($t=5.3$ over five realisations) to $0.004\pm0.025$ ($t=0.33$), i.e.\ to $5\%$ of
its original size, and the ratio of final sizes falls from $1.54$ to $1.02$; at $\beta=0.20$ the
gap falls to $2\%$. Networks with identical static structure spread identically, so the
difference is carried by temporal structure. Targeted surrogates localise it: destroying tie
reinforcement while keeping timestamps and strengths removes about $60\%$ of the gap, and
destroying node inter-event order while keeping its distribution and every pair doubles it, so
the two memories act on spreading with opposite sign in this setting.

\begin{table*}[t]
\caption{\label{tab:regimes}The direction of the node--link difference in spreading depends on
the regime. $c$ is the mean number of contacts per node during one infectious period (twice the
mean initiation rate divided by the recovery rate); $1/\gamma_{\mathrm n}$ and
$1/\gamma_{\mathrm e}$ are the node and tie memory times, in the simulation's time units.
``Node'' or ``link'' names the configuration with the larger outbreak. Details are in the
Supplementary Material, Secs.~S10, S11 and S13.}
\begin{ruledtabular}
\small
\begin{tabular}{llcccl}
Experiment & $N$, activity & $c$ & $1/\gamma_{\mathrm n}$, $1/\gamma_{\mathrm e}$ & $\nnode$ vs $\kappa$ & Larger outbreak \\
\hline
Pure poles [Fig.~\ref{fig:epidemic}(a--d)] & $100$, uniform & ${\approx}29$ & $30$, $600$ & $0.8$ vs $40$ & link (SIR $\epiLink{}$ vs $\epiNode{}$, $\beta=\epiBeta{}$) \\
Matched surrogates (S10) & $200$, uniform & ${\approx}4$ & $30$, $600$ & $0.6$ vs $20$ & node (ratio $1.54$ at $\beta=0.35$) \\
Attribution ridge (S10) & $50$ and $13$, recorded & ${\approx}10$ & $30$, $600$ & $0.75$ vs $30$, $0.8$ vs $40$ & node (school $0.27$ vs $0.15$) \\
Small-system poles (S11) & $55$, lognormal & $0.8$--$17$ & $1$, $1$ & $0.85$ vs $40$ & link, at every $c$ \\
Poles at scale (S13) & $500$--$10^4$, lognormal & ${\approx}3.3$ & $1$, $1$ & $0.85$ vs $40$ & node; link $\approx$ Poisson \\
Fitted vs recorded (S11) & $55$ and $13$, recorded & matched & fitted & fit vs $\nnode=0$ & recorded (fit up to $2\times$ low) \\
\end{tabular}
\end{ruledtabular}
\end{table*}

\paragraph*{The direction depends on the regime.} Table~\ref{tab:regimes} collects every
spreading experiment in the paper. The size of the node--link difference is large in all of
them, from tens of per cent to a factor of several, but its direction is not fixed. It changes
with the number of contacts per infectious period, with system size and with the memory
timescales. When contacts per infectious period are few, node bursts can pack several contacts
into one infectious period and carry the outbreak; when they are many, the long quiet gaps of
a bursty node waste infectious time. In the at-scale ensembles the tie memory decays within one
time unit, much faster than the typical interval between a node's contacts (${\approx}18$ units),
so tie reinforcement there has little dynamical effect and the link pole behaves like a Poisson
sequence. We do not attempt to map the boundary between these regimes. For identifiability the
point is the opposite one: because the direction depends on details the coarse statistics do not
fix, the effect of memory on an epidemic cannot be predicted without knowing which carrier the
memory sits on.

\paragraph*{A real record.} The practical cost is measurable on data. Fitting the high-school
record returns $\hat{n}_{\mathrm{node}}\approx0.9$, the inter-event artefact of Sec.~\ref{sec:plane}, and a
large $\hat\kappa$. Simulating SIR from this fitted model, at the same number of contacts as the
real record, under-predicts the outbreak on the recorded contact sequence by up to a factor of
$\fitUnderFold{}$ across transmissibilities, whereas a model that sets $\nnode=0$ and keeps $\hat\kappa$ is $0.67$--$1.02$ of the recorded outbreak for $\beta\le0.2$ and overshoots it by $1.5$--$1.7$ times for $\beta\ge0.3$. The under-prediction is not general: on the baboon record the naive fit is within $6\%$ of the recorded outbreak for $\beta\ge0.004$ and $16\%$ low at $\beta=0.002$, where the outbreak is small (Supplementary Material, Sec.~S11). At a fixed number of contacts, the spurious node bursts cluster contacts in time and
reach fewer distinct susceptible nodes. This is not an irreducible ambiguity, since $\kappa$ is
identifiable and the artefactual $\hat{n}_{\mathrm{node}}$ can be diagnosed, but it is the price of mistaking
one for the other.

\paragraph*{At scale.} In larger systems (up to $N=10^4$ and ${\sim}10^6$ contacts; last
setting of Table~\ref{tab:regimes}) the difference becomes a shift of the epidemic threshold.
Under SIR the node pole produces a macroscopic outbreak while the tie-reinforced and Poisson
sequences of the same activity stay small, and the gap widens between $N=90$ and $N=10^3$ and then stays near $0.16$. Under SIS the
node pole becomes endemic at a lower transmissibility ($\beta_c\!\approx\!0.14$) than the other
two ($\beta_c\!\approx\!0.30$). Finite-size diagnostics locate this shift but do not resolve a
sharp critical point: the intensive susceptibility collapses across sizes
[Fig.~\ref{fig:epidemic}(e)], the quasi-stationary susceptibility grows sub-linearly, and Binder
cumulants do not cross at a common point (Supplementary Material, Sec.~S13).

\begin{figure*}[t]
  \centering
  \includegraphics[width=\textwidth]{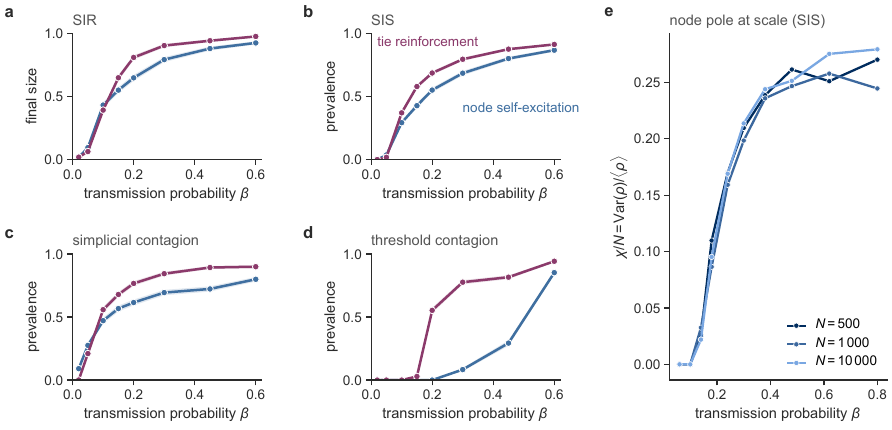}
  \caption{\textbf{Spreading on the two pure poles.} Node self-excitation ($\kappa=0$) and tie
  reinforcement ($\nnode=0$), matched on activity and base attractiveness, against the
  per-contact transmission probability $\beta$ (bands are standard errors over realisations).
  \textbf{(a)}~SIR final size and \textbf{(b)}~SIS endemic prevalence (simple contagion),
  \textbf{(c)}~simplicial and \textbf{(d)}~threshold contagion (complex contagion), all in the
  setting of the first row of Table~\ref{tab:regimes}. \textbf{(e)}~Intensive SIS susceptibility
  $\chi/N=\mathrm{Var}(\rho)/\langle\rho\rangle$ on the node pole for $N=500$--$10^4$ (last
  setting of Table~\ref{tab:regimes}): the curves collapse, so the node pole's lower endemic
  threshold is real but the growth of $\chi$ is the size prefactor, not a diverging
  susceptibility (Supplementary Material, Sec.~S13). The two poles,
  indistinguishable in the coarse statistics, produce epidemics that differ by tens of per cent
  under every dynamics.}
  \label{fig:epidemic}\label{fig:dynpoles_main}
\end{figure*}

\section{Discussion}
\label{sec:discussion}

\subsection{What does the separation result establish, and what does it leave open?}

Proposition~\ref{prop:orthogonality} is a statement about a model class, not about
nature. It says that within the family of Sec.~\ref{sec:model}, and given initiator
labels, node memory and tie memory are estimated from disjoint sufficient statistics and
therefore cannot be traded against one another. That locates the origin of cross-channel
confounding: in face-to-face data, it is the \emph{missing initiator label}. Instrumentation
that recorded who initiated, even imperfectly or for a subsample, would decouple the two
channels in a way that no amount of additional undirected data can; by
Corollary~\ref{cor:dichotomy} it would also remove the upward bias of $\hat\kappa$.

Resolving the cross-channel coupling does not, however, make node memory identifiable, and
the two limits should be kept apart. The first is structural: without the initiator label the
two blocks are coupled. When the data carry exponential node memory (as in our
synthetic benchmark with $\confSynN{}$ nodes) this appears as a confound that survives
posterior smoothing, the estimates trading off along a thin ridge (correlation
$\confSynSmooth{}$), and recording the initiator removes it. The second is empirical and does
not depend on labels: on records with heavy-tailed timing, the likelihood ties $\hat{n}_{\mathrm{node}}$ to
the inter-event-time distribution, whether or not the initiator is observed. It appears in
fully directed human and algorithmic records (email, Stack Overflow, bots, microservices) as
much as in proximity data. Recording the initiator isolates $\kappa$; it does not rescue
$\nnode$ from the inter-event-time distribution.

A practitioner gains three things from these results, which are negative on their face. First,
a two-part \emph{test} to run before a fitted $\hat{n}_{\mathrm{node}}$ is reported: the inter-event-order
shuffle shows whether the value carries information beyond the waiting-time distribution, and the
model-free memory coefficient shows whether there is node memory to measure. Second, an
\emph{estimator fix}, the hierarchical prior on the baseline activities, which restores honest
coverage when nodes have few events. Third, from the attribution principle, an \emph{a priori}
map of which mechanisms a given instrumentation can recover: recording who initiated is the
decisive intervention for separating tie reinforcement from node-level dynamics, whereas
recovering heavy-tailed node dynamics requires access to internal node states or external
drivers rather than contact timings alone.

The result does not say that the model class is right, that $\nnode$ and $\kappa$ are the
only mechanisms operating, or that a mechanism outside the class could not mimic either.
This is why the goodness-of-fit analysis (Sec.~\ref{sec:res_gof}) carries the empirical weight:
the diagnostics reject the model exactly where the known misspecifications sit (group contacts
recorded as simultaneous dyads, and heavy-tailed inter-event laws with only weak order), which licenses
the negative reading rather than a confident positive one.

\subsection{How does the parametrisation relate to existing mechanism families?}

The node self-excitation term reduces, at $\kappa=0$, to a continuous-time
activity-driven model with non-Markovian node activity, close in spirit to the
memory extensions of the activity-driven class~\cite{Karsai2014,Ubaldi2017}, to
self-excitement mechanisms for the activity itself~\cite{Zino2018}, and to work
on aging in temporal networks~\cite{Tizzani2018}. The tie-reinforcement term at
$\nnode=0$ is a choice-model analogue of edge-level memory~\cite{Vestergaard2014} and of
dyadic self-exciting processes~\cite{Butts2008}. The contribution here is not either
limit but the joint parametrisation, which makes ``how much of each'' a well-posed
estimation question rather than a modelling preference.

Williams \emph{et al.}~\cite{Williams2022NC,Williams2022PRE} characterise the
multidimensional structure of memory in temporal networks and, in particular, how a
link's activity depends on the past of other links. That work and ours address adjacent
questions: theirs is descriptive and concerns the shape of memory, ours is inferential
and concerns which carrier the memory sits on. A direct comparison on the same
datasets would be informative and is not attempted here.

The identifiability structure is not specific to social contacts. The same
background-versus-self-excitation ambiguity, and the same incidental-parameter degradation
when a per-unit baseline rate is profiled from few events, arise wherever a self-exciting
process is fitted with a heterogeneous immigration rate: in the epidemic-type aftershock
(ETAS) models of seismology, where separating background seismicity from triggered
aftershocks is a long-standing difficulty, and in neural spike-train inference, where tonic
drive and network-driven bursting play the roles of baseline activity and self-excitation.
The attribution principle applies there as well: a mechanism stays orthogonally recoverable
when its score does not depend on the missing attribution label, so we expect both the shuffle
diagnostic and the hierarchical-prior remedy to transfer to those settings.

\subsection{Limitations}

\paragraph*{Exponential kernels.} The kernel we estimate is exponential, chosen because it
admits the $O(1)$ recursion of Eq.~\eqref{eq:recursion}. Empirical inter-event distributions are
heavy-tailed, and aging effects in particular are poorly represented by a single exponential. A
sum of exponentials on a fixed logarithmic grid of rates keeps the recursion at proportionate
cost and approximates a power law over any finite range. Fitted to the contact timing, such a
mixture fits substantially better than a single exponential (timing log-likelihood gain
$+\sumExpHospGain{}$ on the hospital ward and $+\sumExpSchoolGain{}$ on the high school) and
places the node memory on two to three timescales of minutes, so the single exponential is
indeed misspecified. Yet the branching ratio barely moves ($\nnode: \sumExpSchoolSingle{}\!\to\!
\sumExpSchoolMix{}$ for the school, $\sumExpHospSingle{}\!\to\!\sumExpHospMix{}$ for the
hospital). The high fitted value is therefore not produced by the shape of the kernel; it is
produced by the inter-event-time distribution, which any kernel must reproduce.

\paragraph*{Approximate E-step.} The latent-initiator estimator uses sequential
imputation, a filtering approximation to the label posterior (Sec.~\ref{sec:em}). On directed
data its bias can be measured directly: fitting the same events with the exact
observed-initiator likelihood and with sequential imputation, the imputed $\hat{n}_{\mathrm{node}}$ is biased
by $\eStepBiasSyn{}$ on synthetic data of known truth and $\eStepBiasSnn{}$ on the spiking
network ($\snnObsNnode{}\!\to\!\snnLatNnode{}$), and by more than $+100\%$ on the email log, where
tied timestamps compound the approximation ($\emObsNnode{}\!\to\!\emLatNnode{}$). The bias is
real, dataset dependent and can be large, which is a further reason why we report
latent-initiator estimates only as coordinates in the mechanism plane. The
Metropolis-within-Gibbs smoother (Supplementary Material, Sec.~S3) corrects part of this bias
and confirms that the residual confound is information loss; when the initiator is
recorded the bias does not arise.

\paragraph*{Instantaneous contacts and higher-order structure.} Sessionisation converts a
proximity record into instantaneous events and discards duration, which is real information
and plausibly mechanism-relevant. A contact duration is a \emph{pair-attributed} weight, so by
Corollary~\ref{cor:dichotomy} it stays identifiable under a latent initiator (given its own
memory timescale); it would enrich the tie channel rather than resolve the node one. A native
higher-order encoding of group interactions removes the simultaneous-dyad atom by
construction, but our diagnostics show that it does not repair the residual misfit: the
serial-correlation defect and the high $\hat{n}_{\mathrm{node}}$ survive it (Supplementary Material,
Sec.~S3). Both extensions are left to future work.

\paragraph*{Dataset construction.} Two records need a construction step. The cortical
recording gives spike times but no targets; its node-timing statistics use the spike trains
directly, and no selection fit is made. Bot accounts are identified by the
Wikimedia naming convention, checked against the record of bot flags (Sec.~\ref{sec:data});
this rule misses unflagged automated editing under human-named accounts.

\subsection{Outlook}

The identifiability question addressed here is a prerequisite for, rather than an
alternative to, dynamical questions about temporal networks. If node-driven and
link-driven memory cannot be distinguished from contact data, then statements about how
memory affects spreading are statements about a modelling choice. The separation result
shows that they can be distinguished in principle, and identifies the missing observable,
event attribution, that obstructs it in practice.

A concrete extension is to neural wiring. Recent functional connectomics finds that cortical
neurons with similar response properties are preferentially connected, a \emph{probabilistic}
``like-to-like'' rule that refines, rather than confirms, the deterministic slogan that neurons
which fire together wire together~\cite{Ding2025,MICrONS2025}. The attribution principle
suggests how to test such a rule on spike sequences. A symmetric, pair-attributed similarity
$s_{ij}$ has a label-free memory state, but if it enters the target-selection probability it
inherits the same coupling through the normalisation over the initiator's alternatives as
$\kappa$ (Corollary~\ref{cor:dichotomy}); it is fully orthogonal to node timing only when the
presynaptic neuron is known, as it is in spike recordings with identified connections. Node
self-excitation driven by co-firing is label-dependent and, on heavy-tailed data, tied to the
inter-event distribution at short range, as we find on the cortical spike trains. Testing
this requires adding a functional-similarity term to the selection model and fitting it to
spike sequences, which we leave to future work.

\section{Methods}
\label{sec:methods}
\subsection{Model}
\label{sec:model}

\subsubsection{Events and mechanisms}

Let $V$ be a set of $N$ nodes. We model the contact sequence as a marked point process
on $[0,T]$ whose events are tuples
\begin{equation}
\begin{split}
  &(t_e, i_e, j_e, w_e), \qquad e = 1,\dots,E, \\
  &0 < t_1 \le \dots \le t_E \le T ,
\end{split}
\end{equation}
where $t_e$ is the event time, $i_e \in V$ is the \emph{initiator}, $j_e \in V
\setminus \{i_e\}$ the target, and $w_e \in \mathcal W$ an optional \emph{interaction
weight} (mark), such as a duration, intensity, or count carried by the event, with
$w_e \equiv \varnothing$ recovering the unweighted model. The contact itself is undirected;
the initiator label records which of the two endpoints generated the event. Section~\ref{sec:latent} treats the case, relevant to all proximity data, in which the label
is not observed.

Let $\hist_{t^-}$ denote the history strictly before $t$, and write
\begin{equation}
\begin{split}
  H_i(t) &= \{ t_k < t : i_k = i \}, \\
  H_{ij}(t) &= \{ t_k < t : \{i_k,j_k\} = \{i,j\} \}
\end{split}
\end{equation}
for the set of past times at which $i$ initiated an event, and at which the unordered
pair $\{i,j\}$ was in contact, respectively. Note that $H_i$ is defined by the
\emph{initiator} label whereas $H_{ij}$ is label-free; this asymmetry is what drives the
identifiability structure of Sec.~\ref{sec:identifiability}.

The model contains three mechanisms.

\paragraph*{(A0) Memoryless node activity.} Each node carries a baseline initiation rate
$a_i \ge 0$. Taken alone this is the activity-driven model of Perra
\emph{et al.}~\cite{Perra2012} in continuous time.

\paragraph*{(A1) Node self-excitation.} A node that has recently acted is more likely to
act again, which generates burstiness in the timing of a node's events. The node
initiation intensity is the Hawkes-type~\cite{Hawkes1971} conditional intensity
\begin{equation}
\begin{split}
  \rho_i(t \mid \hist_{t^-}) &= a_i + \nnode \!\!\sum_{t_k \in H_i(t)} \!\! \phi_{\mathrm{n}}(t - t_k), \\
  \phi_{\mathrm{n}}(\tau) &= \gnode e^{-\gnode \tau} .
\end{split}
  \label{eq:rho}
\end{equation}
The kernel is normalised to unit integral, $\int_0^\infty \phi_{\mathrm{n}} = 1$, so
that $\nnode$ is the \emph{branching ratio}: the expected number of events directly
triggered by one event. It is dimensionless, and stationarity of the process is exactly
$\nnode < 1$.

\paragraph*{(B) Tie reinforcement.} Having interacted with a given target recently makes
that target more likely to be chosen again, which is memory carried by the link rather
than the node~\cite{Karsai2014,Vestergaard2014,Ubaldi2017}. Given that $i$ initiates an
event at $t$, the target is drawn from
\begin{equation}
\begin{split}
  \pi_{j \mid i}(t) &= \frac{b_j + \kappa \, m_{ij}(t)}{Z_i(t)}, \\
  Z_i(t) &= \sum_{l \neq i} \bigl( b_l + \kappa \, m_{il}(t) \bigr),
\end{split}
  \label{eq:pi}
\end{equation}
with tie memory
\begin{equation}
  m_{ij}(t) = \sum_{t_k \in H_{ij}(t)} e^{-\gedge (t - t_k)} .
  \label{eq:m}
\end{equation}
Here $b_j \ge 0$ is a node attractiveness and $\kappa \ge 0$ the reinforcement strength.

\paragraph*{(C) Interaction weight (optional mark).} When events carry a weight $w_e$, it is
drawn from a mark law $\pi_{\mathrm W}(w_e \mid \hist_{t_e^-}; \thW)$ that contributes a third
parameter block $\thW$. Two structurally distinct regimes matter for identifiability,
mirroring the label-defined ($H_i$) versus label-free ($H_{ij}$) asymmetry above. A
\emph{pair-attributed} weight depends only on the unordered pair and the label-free tie
history $H_{ij}$ (for instance a contact's duration); an \emph{initiator-attributed} weight
depends on the ordered initiator and its own history $H_i$ (for instance a node's intensity
of engagement when it initiates). As Sec.~\ref{sec:identifiability} shows, it is this
distinction, not the presence of a weight, that decides whether the mark inherits the
initiator confound.

\begin{figure*}[t]
  \centering
  \includegraphics[width=\textwidth]{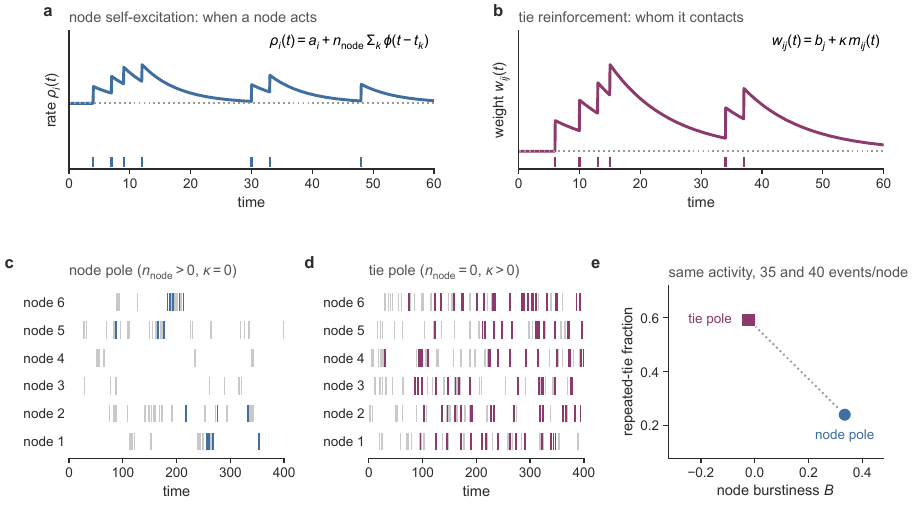}
  \caption{\textbf{The two memories act on orthogonal features of an event, and leave orthogonal signatures.} \emph{(a)}~Node self-excitation raises \emph{when} a node acts: the initiation intensity $\rho_i(t)$ of Eq.~\eqref{eq:rho} jumps after each event and decays, producing temporal bursts. \emph{(b)}~Tie reinforcement raises \emph{which} target is chosen: the selection weight $b_j+\kappa\,m_{ij}(t)$ of Eq.~\eqref{eq:pi} climbs each time a tie is used, so the same tie recurs. \emph{(c,d)}~Event rasters of the two pure poles with matched baseline activity and attractiveness, so that the per-node event count is nearly the same (35 and 40 events per node). Each tick is one initiation of a node; dark ticks are contacts with a partner that node has already contacted, light grey ticks are first contacts. At the node pole ($\nnode>0,\ \kappa=0$) each node's events cluster in time but scatter across partners; at the link pole ($\nnode=0,\ \kappa>0$) they are spread in time but recur with a few partners. \emph{(e)}~The two poles have nearly the same activity yet separate on orthogonal structural axes, node burstiness $B$ and the repeated-tie fraction. This orthogonality is the geometric content of Proposition~\ref{prop:orthogonality}.}
  \label{fig:dynamics}
\end{figure*}

\subsubsection{Normalisations, and why they are not cosmetic}
\label{sec:normalisation}

Two normalisations are imposed, and both are required for the mechanism parameters to
mean anything.

First, the node kernel in Eq.~\eqref{eq:rho} has unit \emph{integral}. This is what
makes $\nnode$ a branching ratio rather than an amplitude with units of inverse time,
and it is what makes the stability condition read $\nnode<1$ independently of $\gnode$.

Second, the tie kernel in Eq.~\eqref{eq:m} has unit \emph{peak}: a target contacted an
instant ago contributes exactly $\kappa$ of excess weight against a baseline of $b_j$.
Combined with the constraint
\begin{equation}
  \frac{1}{N}\sum_{j \in V} b_j = 1 ,
  \label{eq:bnorm}
\end{equation}
this makes $\kappa$ directly interpretable: a just-contacted target is $1+\kappa$ times
as attractive as an average one. Constraint~\eqref{eq:bnorm} is not optional. The
selection probability~\eqref{eq:pi} is invariant under
$(\bm{b}, \kappa) \mapsto (c\,\bm{b}, c\,\kappa)$ for any $c>0$, so without a fixed
normalisation of $\bm{b}$ the parameter $\kappa$ is defined only up to an arbitrary
scale.

This point is the reason for the present parametrisation. A natural-looking alternative,
used in earlier formulations of the problem, is to write a dyadic intensity as a convex
mixture of a node-driven and a link-driven term,
\begin{equation}
\begin{split}
  \lambda_{ij}(t) &= \alpha \, a_i a_j \\
    &\quad + (1-\alpha) \Bigl( \mu + \sum_{t_k \in H_{ij}(t)} \varphi(t-t_k) \Bigr),
\end{split}
  \label{eq:badmodel}
\end{equation}
and to read $\alpha$ as the mechanism weight. Equation~\eqref{eq:badmodel} is not
identified: the product $a_i a_j$ has a free overall scale, and $a_i \mapsto c\, a_i$ is
absorbed exactly by $\alpha \mapsto \alpha / c^2$. If the activities are estimated by a
plug-in moment rule and then held fixed (the usual practice), $\hat\alpha$ reports
the arbitrary normalisation of that rule rather than a property of the data. In our
numerical experiments, rescaling the supplied activities by a factor of two moved
$\hat\alpha$ from the ``mixed'' regime to the ``strongly link-driven'' regime on
\emph{unchanged} data. In the present model no analogous freedom exists: the activities
$\{a_i\}$ are profiled out of the likelihood rather than plugged in
(Sec.~\ref{sec:profiling}), and $\nnode$ and $\kappa$ are dimensionless ratios pinned by
the two kernel normalisations and by Eq.~\eqref{eq:bnorm}.

A second defect of Eq.~\eqref{eq:badmodel} is structural rather than numerical. Its
node-driven pole, $\lambda_{ij}=\alpha a_i a_j$, is an inhomogeneous Poisson process with
\emph{independent dyads}: there is no node activation event, and therefore none of the
correlation among a node's simultaneous contacts that defines activity-driven dynamics.
Moreover its A pole is memoryless while its B pole carries memory, so a single parameter
$\alpha$ conflates the node-versus-link axis with the memoryless-versus-memory axis.
Evidence for $\alpha<1$ is then evidence for memory of some kind, not evidence that the
memory resides on links. Separating $\nnode$ from $\kappa$, as in
Eqs.~\eqref{eq:rho}--\eqref{eq:m}, removes both defects: node burstiness and tie
reinforcement are distinct parameters that can be nonzero independently.

\subsubsection{Induced dyadic intensities and stationary rate}

The directed intensity of contacts initiated by $i$ towards $j$ is
\begin{equation}
  \lambda_{i \to j}(t) = \rho_i(t)\, \pi_{j \mid i}(t) ,
\end{equation}
and the intensity of the undirected contact $\{i,j\}$ is
$\lambda_{ij} = \lambda_{i \to j} + \lambda_{j \to i}$.

Because each node's ground process is a linear Hawkes process with immigration rate
$a_i$ and branching ratio $\nnode$, the stationary expected rate of node $i$ is
$a_i/(1-\nnode)$ and the total event rate is
\begin{equation}
  \Lambda = \frac{1}{1-\nnode} \sum_{i \in V} a_i ,
  \qquad \nnode < 1 .
  \label{eq:rate}
\end{equation}
Equation~\eqref{eq:rate} is a useful implementation check: it holds only if the node
kernel is unit-integral, so a simulator that violates it has a normalisation error. The
target-selection mechanism does not enter Eq.~\eqref{eq:rate}, since $\kappa$ redistributes
contacts among targets without changing how often a node acts. This is a first hint of
the separation established in Sec.~\ref{sec:identifiability}.

\subsection{Likelihood}
\label{sec:likelihood}

\subsubsection{Complete-data log-likelihood}

Write $\thT = (\bm{a}, \nnode, \gnode)$, $\thS = (\bm{b}, \kappa, \gedge)$ and, when a weight is
present, $\thW$ for the timing, selection and weight parameters. For a marked point process
observed on $[0,T]$ with ground intensity $\rho_i$, target law $\pi_{j\mid i}$ and weight law
$\pi_{\mathrm W}$, the log-likelihood is~\cite{DaleyVereJones,Ogata1978}
\begin{equation}
\begin{split}
  \loglik(\thT,\thS,\thW) &= \sum_{e=1}^{E} \Bigl[ \log \rho_{i_e}(t_e) + \log \pi_{j_e \mid i_e}(t_e) \\
  &\qquad\quad + \log \pi_{\mathrm W}(w_e) \Bigr]
   - \sum_{i \in V} \int_0^T \!\! \rho_i(s)\, ds ,
\end{split}
  \label{eq:loglik}
\end{equation}
all intensities being evaluated at $t_e^-$, i.e.\ conditionally on the strict history; the
weight term is present only for marked events, and dropping it recovers the unweighted
likelihood.

The compensator is available in closed form. Using
$\int_{t_k}^{T} \gnode e^{-\gnode(s-t_k)} ds = 1 - e^{-\gnode(T-t_k)}$,
\begin{equation}
\begin{split}
  \int_0^T \!\! \rho_i(s)\, ds &= a_i T + \nnode\, C_i(\gnode), \\
  C_i(\gnode) &= \!\!\sum_{t_k \in H_i(T)} \!\! \bigl( 1 - e^{-\gnode (T-t_k)} \bigr).
\end{split}
  \label{eq:compensator}
\end{equation}
Substituting Eqs.~\eqref{eq:rho}, \eqref{eq:pi} and \eqref{eq:compensator} into
Eq.~\eqref{eq:loglik} and introducing the node and dyad excitation states
\begin{align}
  S_i(t) &= \!\!\sum_{t_k \in H_i(t)} \!\! \gnode e^{-\gnode(t-t_k)}, \notag \\
  M_i(t) &= \sum_{l \neq i} m_{il}(t),
  \label{eq:states}
\end{align}
gives the decomposition on which everything else rests:
\begin{equation}
  \boxed{\;\loglik(\thT,\thS,\thW) = \logT(\thT) + \logS(\thS) + \logW(\thW)\;}
  \label{eq:split}
\end{equation}
with
\begin{align}
  \logT(\thT) &= \sum_{e=1}^{E} \log\!\bigl( a_{i_e} + \nnode S_{i_e}(t_e^-) \bigr) \notag \\
   &\quad - \sum_{i \in V} \bigl[ a_i T + \nnode C_i(\gnode) \bigr],
  \label{eq:logT}\\[2pt]
  \logS(\thS) &= \sum_{e=1}^{E} \Bigl[
      \log\!\bigl( b_{j_e} + \kappa\, m_{i_e j_e}(t_e^-) \bigr) \notag \\
      &\quad - \log Z_{i_e}(t_e^-) \Bigr],
  \label{eq:logS}
\end{align}
where $Z_i(t) = (B - b_i) + \kappa M_i(t)$ and $B = \sum_l b_l = N$ by
Eq.~\eqref{eq:bnorm}. When events are marked, a third block evaluated on the realised
weights alone,
\begin{equation}
  \logW(\thW) = \sum_{e=1}^{E} \log \pi_{\mathrm W}\bigl( w_e \mid \hist_{t_e^-}; \thW \bigr),
  \label{eq:logW}
\end{equation}
joins the decomposition, depending on the labels only through whichever history ($H_i$ or
$H_{ij}$) its attribution selects. Equation~\eqref{eq:split} is thus additively separable
across all three blocks, each a function of its own parameters and its own sufficient
statistics, which is the property the identifiability results turn on.

\subsubsection{Linear-time evaluation}

Both $S_i$ and $m_{ij}$ obey a one-step recursion, since for exponential kernels
\begin{equation}
\begin{split}
  S_i(t_{e+1}^-) &= e^{-\gnode \Delta_e}\bigl[ S_i(t_e^-) + \gnode \,\delta_{i,i_e} \bigr], \\
  \Delta_e &= t_{e+1}-t_e ,
\end{split}
  \label{eq:recursion}
\end{equation}
and analogously for $m_{ij}$ and $M_i$ with rate $\gedge$ and unit increments. Evaluating
Eqs.~\eqref{eq:logT}--\eqref{eq:logS} therefore costs $O(E)$ rather than the
$O(\sum_i K_i^2)$ of a direct kernel resummation, where $K_i=|H_i(T)|$. This is what
makes fitting the full contact record feasible; earlier treatments of this problem
truncated the data to a few thousand events for exactly this reason, which is both
unnecessary and, because the truncation selects a contiguous stretch of one day, a
source of daily-cycle bias.

\subsubsection{Unboundedness under coincident events}
\label{sec:coincident}

Empirical proximity data is recorded on a discrete grid (\samplingGrid{}~s for
SocioPatterns), so many events share a timestamp exactly, which
makes the likelihood unbounded.

\begin{proposition}[Likelihood divergence on tied timestamps]
\label{prop:divergence}
Suppose there exist a node $i$ with $a_i > 0$ and a time $t$ at which $i$ initiates at
least two events, and let $\nnode > 0$ be fixed. Then
$\sup_{\gnode > 0} \logT(\thT) = +\infty$.
\end{proposition}

\begin{proof}
Order the coincident events at $t$ and consider the second one, $e^\star$. By
Eq.~\eqref{eq:states} the first contributes to $S_i$ at zero lag, so
$S_i(t_{e^\star}^-) \ge \gnode$ and the corresponding term in Eq.~\eqref{eq:logT} is at
least $\log(a_i + \nnode \gnode) \to +\infty$ as $\gnode \to \infty$. Every other event term is
bounded below by $\log a_{i_e}$, a constant independent of $\gnode$, since $S \ge 0$. The
compensator is bounded above: $C_i(\gnode) \le K_i$ for all $\gnode$, so
$\sum_i [a_i T + \nnode C_i] \le T\sum_i a_i + \nnode E$, a constant. Hence
$\logT \to +\infty$.
\end{proof}

The practical consequence is that an unconstrained optimiser drives $\gnode$ to its upper
bound and reports a spurious fit with an apparently excellent likelihood. We therefore
restrict the decay rates to
\begin{equation}
\begin{split}
  \frac{1}{10\,T} &\;\le\; \gnode,\ \gedge \;\le\; \frac{1}{\Delta_{\min}}, \\
  \Delta_{\min} &= \min_{e:\, \Delta_e > 0} \Delta_e ,
\end{split}
  \label{eq:gammabounds}
\end{equation}
where $\Delta_{\min}$ is the smallest strictly positive inter-event gap. The upper bound
is a statement about resolution, not a numerical expedient: a memory kernel decaying
faster than the sampling grid is not estimable from data on that grid. The lower bound
excludes kernels flat over the observation window, which are indistinguishable from a
shift in the baseline $a_i$.

\begin{remark}
Proposition~\ref{prop:divergence} is not specific to our parametrisation. Any
self-exciting point process fitted to grid-sampled data with an unbounded decay rate
suffers from it. It is worth stating because the failure is silent: the reported
log-likelihood improves monotonically as the fit becomes meaningless.
\end{remark}

\subsubsection{Latent initiator}
\label{sec:latent}

Proximity sensors record an undirected contact and no initiator label. The observed data
are then $c_e = (t_e, \{u_e, v_e\})$, and the complete data are recovered by a label
$\sigma_e \in \{0,1\}$ with $i_e = u_e$ if $\sigma_e = 0$ and $i_e = v_e$ otherwise. The
observed-data likelihood is the marginal
\begin{equation}
  L(\thT,\thS) = \sum_{\bm{\sigma} \in \{0,1\}^{E}}
    \exp \loglik\bigl( \thT, \thS \,;\, \bm{c}, \bm{\sigma} \bigr) .
  \label{eq:marginal}
\end{equation}
Equation~\eqref{eq:marginal} does \emph{not} factorise over events. The dyad memories
$m_{ij}$ and their row sums $M_i$ are label-free, but the node states $S_i$, the
compensators $C_i$, and the choice of normaliser $Z_{i_e}$ all depend on the labels of
all earlier events. The sum therefore has $2^E$ terms with no product structure, and
must be approximated; Sec.~\ref{sec:em} describes the Monte Carlo EM scheme used here
and states plainly what it does and does not guarantee.

\subsection{Identifiability}
\label{sec:identifiability}

\subsubsection{An attribution--orthogonality principle}

The decomposition~\eqref{eq:split} is an instance of a general principle, and stating it in
general clarifies what is and is not special about our model. What governs identifiability is
neither directedness nor weighting but whether the per-event \emph{attribution labels} (which
endpoint initiated, and, for a weight, which endpoint's process carries it) are observed. We
state the principle for an arbitrary number of mechanism blocks and then read off the model of
this paper as the case at hand.

\begin{proposition}[Attribution orthogonality]
\label{prop:orthogonality}
Let the complete-data log-likelihood be additively separable across $B$ variation-independent
parameter blocks, $\loglik_{\mathrm c}(\bm\theta)=\sum_{b=1}^{B}\loglik_b(\bm\theta_b)$ with
$\bm\Theta=\prod_b\bm\Theta_b$, where each $\loglik_b$ may depend on latent per-event
attribution labels $Z$ through its own sufficient statistics but on no other block's
parameters. Then:
\begin{enumerate}
\item[\emph{(i)}] \emph{Observed attribution.} If $Z$ is observed, the Fisher information is
block diagonal, $\mathcal I(\bm\theta)=\mathrm{diag}(\mathcal I_1,\dots,\mathcal I_B)$; the
block maximum-likelihood estimators are asymptotically independent and each block's profile
likelihood is invariant to the others.
\item[\emph{(ii)}] \emph{Latent attribution.} If a sub-vector $z^{\mathrm{mis}}$ of the labels
is marginalised, Louis' identity gives $\mathcal I_{\mathrm{obs}}=\mathcal I_{\mathrm c}
-\E_O\!\big[\Cov_{z^{\mathrm{mis}}\mid O}(\nabla\loglik_{\mathrm c})\big]$, whose $(b,b')$
off-diagonal block is $-\E_O\!\big[\Cov_{z^{\mathrm{mis}}\mid O}(\nabla_{\bm\theta_b}\loglik_b,\,
\nabla_{\bm\theta_{b'}}\loglik_{b'})\big]$. Hence a block that is \emph{label-free} in
$z^{\mathrm{mis}}$ stays orthogonal to every other block; two blocks that both depend on
$z^{\mathrm{mis}}$ generically confound, the coupling vanishing as the posterior concentrates
($\Ent(z^{\mathrm{mis}}\mid O)\!\to\!0$), and under uniform label revelation scaling with the residual
entropy order parameter $(1-\rho)\bar\Ent$.
\end{enumerate}
\end{proposition}

\begin{proof}
(i) By separability the mixed second derivatives $\partial^2\loglik_{\mathrm c}/
(\partial\bm\theta_b\,\partial\bm\theta_{b'}^{\!\top})=0$ for $b\neq b'$, since $\loglik_b$ does
not contain $\bm\theta_{b'}$; taking expectations gives block-diagonal $\mathcal I$, whose
inverse (the asymptotic MLE covariance) is block diagonal. (ii) Louis' identity applied to the
separable score $\nabla\loglik_{\mathrm c}=(\nabla_{\bm\theta_b}\loglik_b)_b$: a label-free
block's score is $\sigma(O)$-measurable, so its conditional covariance with any statistic
vanishes, giving the stated off-diagonal; the vanishing follows because the
conditional covariance degenerates as the posterior concentrates. The detailed computation for the
model of this paper is in the Supplementary Material, Sec.~S9.
\end{proof}

The unweighted model of this paper is the two-block case ($B=2$) with attribution
$z^{\mathrm{mis}}=\{i_e\}$: the timing block (through $H_i$) and the selection block (through
the ordered choice) both depend on the initiator. By
Proposition~\ref{prop:orthogonality}(i) an observed initiator makes them orthogonal
($\mathrm{corr}(\hat{n}_{\mathrm{node}},\hat\kappa)=0$ exactly, verified in Sec.~\ref{sec:res_synthetic}), while by (ii) a
latent initiator couples them, with magnitude determined by the initiator-posterior entropy,
which is the confound documented in Sec.~\ref{sec:res_latent}. In particular the profile likelihood for $\nnode$ does
not depend on $\kappa$, and conversely. Adding the optional weight block gives the weighted
generalisation directly.

\begin{proposition}[Weight identifiability under a latent initiator]
\label{prop:weight}
For the three-block model $\bm\theta=(\thT,\thS,\thW)$ with
$\loglik_{\mathrm c}=\logT+\logS+\logW$:
\emph{(i)} with the initiator observed, the Fisher information is block diagonal in
$(\thT,\thS,\thW)$ under either weight regime, so all three mechanisms are mutually
orthogonal; \emph{(ii)} with the initiator latent and a \emph{pair-attributed} weight
($H_{ij}$), the weight block stays orthogonal and only timing and selection confound, exactly
as in the unweighted model; \emph{(iii)} with the initiator latent and an
\emph{initiator-attributed} weight ($H_i$), all three blocks depend on the same latent label
and the missing information is a full $3\times3$ coupling, so the weight mechanism is
non-identifiable jointly with node timing and tie selection. The magnitude of the entire
coupling is set by one scalar, the initiator-posterior entropy $\Ent(\{i_e\}\mid O)$.
\end{proposition}

\begin{proof}
Apply Proposition~\ref{prop:orthogonality} with $z^{\mathrm{mis}}=\{i_e\}$. Timing and
selection are label-dependent; the weight block is label-free iff it is pair-attributed.
Part~(i) is case~(i) of the principle. For (ii), a label-free weight block has zero
off-diagonal with every block, leaving the timing--selection coupling. For (iii), an
initiator-attributed weight block is label-dependent and so couples to both, producing the
full $3\times3$ block. Scaling with the residual posterior entropy is part~(ii) of the principle.
\end{proof}

\begin{corollary}[Label-free identifiability trichotomy]
\label{cor:dichotomy}
Under a latent initiator, a mechanism stays orthogonal to the others if and only if its
\emph{score} is $\sigma(\bm c)$-measurable. The test applies to the likelihood, not to the
memory state, and the two can disagree. Node self-excitation is label-dependent in both: its
state is the initiator history $H_i$ and its score depends on who acted. A pair-attributed
mark (e.g.\ a contact duration) is label-free in both, and stays orthogonal. Tie reinforcement
falls in between: its memory state is the label-free pair history $H_{ij}$, but its likelihood
is the conditional choice $\pi_{j\mid i}(t)=w_{ij}(t)/\sum_k w_{ik}(t)$, whose normalisation
runs over the initiator's alternatives, so
$\nabla_\kappa\log\pi_{j\mid i}=\nabla_\kappa\log w_{ij}-\nabla_\kappa\log\sum_k w_{ik}$ is not
label-free. A label-free memory state is necessary but not sufficient. Weight attribution
therefore decides whether the mark inherits the confound, while tie reinforcement retains a
residual coupling carried entirely by its normalisation, which is what inflates $\hat\kappa$
under latent labels rather than leaving it clean (Supplementary Material, Sec.~S14).
\end{corollary}

As a concrete case, a contact \emph{duration} whose expected value reinforces on its own
timescale, $\mu_{ij}(t)=d_0+\kappa_d\,m^{\mathrm D}_{ij}(t)$ with $m^{\mathrm D}$ a tie-memory
kernel of rate $\gamma_d$ (distinct from $\gedge$, so the block is variation-independent), is
pair-attributed and therefore identifiable under a latent initiator by
Corollary~\ref{cor:dichotomy}: it measures whether reinforcement acts through interaction
frequency ($\kappa$) or length ($\kappa_d$). It enriches the already-identifiable tie channel
and does not bear on node identifiability.

\begin{remark}[One order parameter for directedness, weighting and marking]
Because the whole coupling is controlled by $\Ent(\{i_e\}\mid O)$, the label-revelation
experiment used for the two-block confound (Results) applies verbatim: revealing a fraction
$\rho$ of true initiators collapses every off-diagonal along the same $(1-\rho)\bar\Ent$ order
parameter. Directedness, weighting and marking are thus not separate identifiability questions
but one: \emph{is the attribution observed?}, which is answered by a single missing-information
object.
\end{remark}

\paragraph*{Numerical confirmation of the weighted case.} We test
Proposition~\ref{prop:weight} on the three-mechanism process ($N=25$,
$\sim\!1.7\times10^3$ events, $\nnode=0.5$, $\kappa=10$) with a mark whose log-scale couples
with strength $c_W$ either to the initiator's self-excitation state
(\emph{initiator}-attributed, $H_i$) or to the $(i,j)$ tie-memory state
(\emph{pair}-attributed, $H_{ij}$), on the same realised network. Masking a fraction $1-\rho$
of the true initiators and imputing them, a numerical experiment (Supplementary Material, Sec.~S9) confirms Corollary~\ref{cor:dichotomy}: for
the initiator-attributed mark the estimator spread $\sigma_{\hat c_W}$ is large under a latent
initiator and collapses as labels are revealed, and its missing-information off-diagonal with
$\hat\kappa$ is substantial, meaning the mark joins the node/tie confound as a full $3\times3$
coupling; for the pair-attributed mark $\sigma_{\hat c_W}\!\approx\!0$ at every $\rho$ and the
off-diagonals vanish, so it stays identifiable regardless of the initiator
(Corollary~\ref{cor:dichotomy}). (This is a proof-of-concept on one realisation family; once
the spread collapses the correlation coefficient is ill-conditioned, so the reported
discriminators are the spread and the covariance, which vanish honestly.)

\begin{remark}[The result is kernel-agnostic]
\label{rem:kernel_agnostic}
The factorisation~\eqref{eq:split}, and hence Proposition~\ref{prop:orthogonality}, uses only
that the node-timing intensity depends on node $i$'s own past while the target law depends on
the interaction history, not the exponential form of the kernel. Replacing
$\phi_{\mathrm n}(\tau)=\gnode e^{-\gnode\tau}$ by \emph{any} admissible timing kernel, such as a power
law, a gamma kernel, or the discrete empirical distribution of inter-event times in the
regular spiking of the neural data, leaves the timing block a function of node histories
alone and the selection block untouched, so the two mechanisms remain orthogonal under an
observed initiator for the whole class. The exponential is the instance we estimate, chosen
for the $O(1)$ recursion of Eq.~\eqref{eq:recursion}. The identifiability theorem is therefore
a property of the model \emph{structure}; the empirical limit of Sec.~\ref{sec:plane} concerns
estimation, not this structure. Whether a richer kernel class turns the smooth crossover of
Fig.~\ref{fig:plane_validation} into a sharp transition in the space of kernel configurations
we probe directly (Supplementary Material, Sec.~S12), simulating exact self-exciting processes
across a physically-grounded family that runs from short-memory (exponential) to scale-free
(power-law / Omori, as in seismicity and economic order flow). It does not: the boundary stays
a size-independent crossover along the kernel axis too. The probe instead exposes a
\emph{practical} limit that the structural theorem cannot see. Scale-free node memory spreads
its correlation across all lags, so its per-lag signature is faint and the single-exponential
estimator is effectively blind to it ($R\!\to\!1$ at any system size); and because
self-excitation couples memory to burstiness, no Hawkes kernel (whether short-memory or scale-free)
reaches the resolvable corner that a burstiness-free surrogate attains. A scale-free-aware
estimator does not remove this limit (Sec.~\ref{sec:plane}): the absence or presence of node
memory in such a system has to be established by model-free statistics.
\end{remark}

The content of Proposition~\ref{prop:orthogonality} is that the node mechanism and the
link mechanism are informed by \emph{disjoint sufficient statistics}. The timing block
is a function of the inter-event structure of each node's initiation sequence; the
selection block is a function of which target was chosen given that an event occurred,
and is invariant to when the events occurred. No amount of node burstiness can be
reinterpreted as tie reinforcement, or the reverse, because neither enters the other's
estimating equation. This is a structural property of the parametrisation of
Sec.~\ref{sec:model}, not a numerical observation about a particular dataset, and it
fails for the mixing form~\eqref{eq:badmodel}, in which a single parameter multiplies both
channels.

Proposition~\ref{prop:orthogonality} should not be read as a claim that the mechanisms
are always easy to estimate. It says the two estimation problems do not interfere; each
can still be poorly determined on its own if the data are short, which is the subject of
Sec.~\ref{sec:coverage}. Two boundary cases are worth recording. If $\kappa = 0$ the
selection likelihood~\eqref{eq:logS} does not depend on $\gedge$ at all, so the tie memory
timescale is unidentified at the pure node-driven pole; likewise $\gnode$ is unidentified
when $\nnode = 0$. Both are ordinary boundary non-identifiabilities of nuisance
timescales under a null, and they matter for the tests of Sec.~\ref{sec:tests}.

\subsubsection{Loss of orthogonality under a latent initiator}

When the label is unobserved the log-likelihood is the logarithm of the
sum~\eqref{eq:marginal}, and no analogue of Eq.~\eqref{eq:split} holds: each
configuration $\bm\sigma$ contributes a product of a timing factor and a selection
factor, but the mixture over $\bm\sigma$ couples them. Concretely, the posterior over
labels depends on both $\rho$ and $\pi$,
\begin{equation}
  \Prob(\sigma_e = 0 \mid \hist_{t_e^-}, \bm\theta)
  \;\propto\; \rho_{u_e}(t_e)\, \pi_{v_e \mid u_e}(t_e),
  \label{eq:posteriorlabel}
\end{equation}
so a parameter change that makes node memory more plausible also reassigns labels, which
in turn changes the tie-memory sufficient statistics. The information matrix acquires
nonzero off-diagonal blocks and the two mechanisms become confounded.

We quantify this with the across-imputation correlation
$\mathrm{corr}(\hat{n}_{\mathrm{node}}, \hat\kappa)$ computed over the label draws of the EM scheme
of Sec.~\ref{sec:em}. Proposition~\ref{prop:orthogonality} guarantees this quantity is
exactly zero when labels are observed, so any departure from zero measures precisely the
mechanism confound induced by the missing label. Section~\ref{sec:res_latent} reports it for
synthetic and empirical data.

\subsubsection{What a point estimate cannot tell you}
\label{sec:coverage}

The question of interest is not the value of $\hat{n}_{\mathrm{node}}$ but whether the data determine
it. We therefore report, in place of a point estimate and a threshold rule, three
quantities.

\paragraph*{Profile likelihood.} With all remaining parameters (including the $N$
nuisance activities) re-optimised at each value,
\begin{equation}
  \loglik_{\mathrm{p}}(\nnode) = \max_{\gnode,\, \bm{a}} \logT(\bm{a},\nnode,\gnode),
\end{equation}
and the $95\%$ interval is
$\{\nnode : 2[\loglik_{\mathrm{p}}(\hat{n}_{\mathrm{node}}) - \loglik_{\mathrm{p}}(\nnode)] \le
\chi^2_{1,0.95}\}$. The \emph{curvature} of this curve is the identifiability statement;
a flat profile means the mechanism is not recoverable however sharp the argmax looks.

\paragraph*{Observed information.} The Hessian of the concentrated negative
log-likelihood at the optimum, inverted to give asymptotic standard errors. Comparing
these against the across-replicate dispersion of the estimates is a check on whether the
asymptotics have taken hold at the available sample size; a large discrepancy is itself a
finding about the data requirement.

\paragraph*{Coverage.} Over independent replicates at known truth, the fraction of
$95\%$ Wald intervals containing the true value. Coverage is the diagnostic that
distinguishes the two failure modes that matter. An estimator that is imprecise but
honest has wide intervals and nominal coverage; an estimator in a degenerate regime has
narrow intervals and coverage far below nominal, i.e.\ it is \emph{confidently wrong}.
Only the second invalidates a mechanism claim. Section~\ref{sec:res_sweep} maps coverage
over $(N, T, \nnode, \kappa)$ and locates the boundary between the two regimes.

\subsection{Estimation}
\label{sec:estimation}

\subsubsection{Exact profiling of the baseline activities}
\label{sec:profiling}

The $N$ baseline activities are nuisance parameters that grow with the system size, so
they are removed by profiling rather than estimated jointly. Fix $(\nnode, \gnode)$. From
Eq.~\eqref{eq:logT},
\begin{equation}
\begin{split}
  \frac{\partial \logT}{\partial a_i}
  &= \!\!\sum_{e:\, i_e = i} \!\! \frac{1}{a_i + \nnode S_i(t_e^-)} \;-\; T , \\
  \frac{\partial^2 \logT}{\partial a_i^2} &< 0 ,
\end{split}
  \label{eq:score_a}
\end{equation}
so $\logT$ is strictly concave and separable in $\bm{a}$, and the profile maximiser is
obtained node by node.

\begin{lemma}[Bracketing]
\label{lem:bracket}
Let $K_i = |\{e : i_e = i\}| \ge 1$. The stationary point of Eq.~\eqref{eq:score_a}, if
it exists in $[0,\infty)$, lies in $[0, K_i/T]$; otherwise the maximiser is the boundary
point $a_i = 0$.
\end{lemma}

\begin{proof}
The score is continuous and strictly decreasing on $(0,\infty)$. At $a_i = K_i/T$ each
summand is at most $T/K_i$, so the score is at most $K_i (T/K_i) - T = 0$. If the score
at $a_i \to 0^+$ is already non-positive (which requires $S_i(t_e^-)>0$ for every event
of $i$, i.e.\ node memory alone accounts for the node's activity), then the maximum is
at the boundary.
\end{proof}

Lemma~\ref{lem:bracket} makes bisection on $[0,K_i/T]$ safe and unconditionally
convergent. In the memoryless limit $\nnode=0$ the solution reduces to the Poisson
maximum-likelihood estimate $\hat a_i = K_i/T$, which we use as an implementation test.

Profiling, rather than plugging in a moment estimate, is what removes the scale
degeneracy discussed in Sec.~\ref{sec:normalisation}: the activities adapt to whatever
$\nnode$ is under consideration, so there is no fixed exogenous normalisation left for a
mechanism parameter to absorb.

The attractiveness $\bm{b}$ is treated analogously. It is either held uniform
($b_j \equiv 1$, appropriate when node popularity is not of interest) or estimated by a
minorise--maximise fixed point for the choice model~\eqref{eq:pi}, renormalised to
satisfy Eq.~\eqref{eq:bnorm} after every sweep. The objective is monitored and the
iteration is stopped if a sweep fails to improve it.

\subsubsection{Optimisation over the mechanism parameters}

After profiling, two two-dimensional problems remain: $(\nnode, \gnode)$ from the timing
block and $(\kappa, \gedge)$ from the selection block, which by
Proposition~\ref{prop:orthogonality} may be solved independently. Both are solved by
bounded quasi-Newton iterations on $(\nnode, \log\gnode)$ and $(\log\kappa, \log\gedge)$,
subject to Eq.~\eqref{eq:gammabounds}, and seeded from a coarse logarithmic grid scan.

The grid scan is not a refinement. If $\gnode$ is initialised at an $O(1)$ value while the
data live on a scale of hours, every memory term underflows, the gradient vanishes, and
the optimiser terminates at the starting point reporting no memory, regardless of the
truth. We observed exactly this failure on empirical data: an initial fit returned
$\hat\kappa \approx 0$ on a record with a repeated-tie fraction of $0.69$. Seeding from
the data-implied range of Eq.~\eqref{eq:gammabounds} removes it.

\subsubsection{Monte Carlo EM for the latent initiator}
\label{sec:em}

The marginal likelihood~\eqref{eq:marginal} is approximated by Monte Carlo EM. Each
sweep draws $M$ label configurations by \emph{sequential imputation}: proceeding forward
through the event list, $\sigma_e$ is sampled from Eq.~\eqref{eq:posteriorlabel}
evaluated with the state accumulated from the labels already drawn, after which the
complete-data estimator of Sec.~\ref{sec:profiling} is applied to each imputed dataset
and the parameters averaged.

Standard errors combine within- and between-imputation variance by Rubin's
rules~\cite{Rubin1987},
\begin{equation}
  \widehat{\mathrm{Var}}(\hat\theta) = \bar{W}
    + \Bigl( 1 + \tfrac{1}{M} \Bigr) B ,
\end{equation}
with $\bar{W}$ the mean within-imputation variance and $B$ the between-imputation
variance, so the reported uncertainty includes the cost of the missing labels.

\paragraph*{What this does not guarantee.} Sequential imputation samples from a
\emph{filtering} approximation to the label posterior: $\sigma_e$ is drawn conditionally
on earlier labels but is never revised in the light of later events, whereas the exact
E-step requires the \emph{smoothing} distribution
$\Prob(\bm\sigma \mid \bm{c}, \bm\theta)$. The scheme is therefore biased by an amount
we do not currently bound. On synthetic data the bias can be isolated by comparing
against the observed-initiator fit on the same realisation, and we report both; the gap
between them is an upper bound on the sum of the approximation error and the genuine
information loss, and we do not attribute it entirely to the latter. To separate the two we
also sample the exact smoothing distribution with a Metropolis-within-Gibbs sampler over the
labels (Supplementary Material, Sec.~S3), which is the check used in Sec.~\ref{sec:res_latent}.

\subsubsection{Tests of the two boundary nulls}
\label{sec:tests}

Two hypotheses are of direct interest: $H_0^{\mathrm{B}}: \kappa = 0$ (no tie
reinforcement; target selection is purely node-driven) and
$H_0^{\mathrm{A1}}: \nnode = 0$ (no node self-excitation; node activity is memoryless).
Both are tested by likelihood ratio.

Both nulls lie on the boundary of the parameter space, since $\kappa \ge 0$ and
$\nnode \ge 0$. The asymptotic null distribution of the likelihood-ratio statistic is
then not $\chi^2_1$ but the chi-bar-squared mixture
$\tfrac12 \chi^2_0 + \tfrac12 \chi^2_1$~\cite{SelfLiang1987}. Referring the statistic to
$\chi^2_1$, as we do, therefore produces $p$-values that are conservative by roughly a
factor of two. We report them uncorrected rather than silently adjusted; the correction
is available if a borderline case ever depends on it. A further complication is that the
corresponding timescale ($\gedge$ under $H_0^{\mathrm{B}}$, $\gnode$ under
$H_0^{\mathrm{A1}}$) is unidentified under the null, which is the Davies
problem~\cite{Davies1987}; with the bounded parameter ranges of
Eq.~\eqref{eq:gammabounds} this affects the calibration of the test but not the
consistency of the estimator, and we verify the operating characteristics by simulation
rather than relying on the asymptotic reference.

\subsubsection{Implementation}

The reference implementation is in Python (NumPy/SciPy) and is the version used for all
numbers reported here. A Julia port drives the large parameter sweeps; the two are
checked against a shared fixture that pins the event sequence and the likelihood values
evaluated on it, and must agree to $10^{-8}$. Since the two languages have different
random number generators, the simulators are verified separately against the analytic
stationary rate, Eq.~\eqref{eq:rate}. Code and fixtures are available at
\url{https://github.com/micheletizzani/temporal-memory-epidemics}.

\subsection{Data and preprocessing}
\label{sec:data}

\subsubsection{Datasets}

We analyse twelve records (Supplementary Table~S1), grouped by whether the initiator of an event
is recorded.

\paragraph*{Proximity records (initiator not recorded).} Three SocioPatterns face-to-face
records, chosen to span different social settings at identical instrumentation, and one animal
record with the same class of sensor:
\begin{itemize}
\item \textbf{High school}~\cite{Mastrandrea2015}: \hsNodes{} students,
  \hsRawSamples{} proximity samples over \hsSpanHours{}~h spanning \hsDays{} school
  days. Strong class structure and a hard daily cycle.
\item \textbf{Conference (SFHH)}~\cite{Genois2018}: \sfNodes{} attendees,
  \sfRawSamples{} samples over \sfSpanHours{}~h. Weak prior social structure, many
  first-time encounters.
\item \textbf{Hospital ward}~\cite{Vanhems2013}: \hoNodes{} individuals (patients,
  nurses, doctors, administrative staff), \hoRawSamples{} samples over
  \hoSpanHours{}~h. Small, role-structured, shift-driven.
\item \textbf{Baboons}~\cite{Gelardi2020}: $13$ baboons with ${\sim}6\times10^{4}$
  co-presence samples. It tests whether the conclusions hold in a different species and social
  system, and it is the second record of the forecast comparison (Supplementary Material,
  Sec.~S11).
\end{itemize}
The three SocioPatterns records are sampled on the same \samplingGrid{}~s grid by the same RFID
platform, so the preprocessing below applies uniformly.

\paragraph*{Communication records (sender is the initiator).}
\begin{itemize}
\item \textbf{Email}~\cite{Paranjape2017}: the directed email log of a European research
  institution, \gofEmContacts{} messages among \gofEmN{} active accounts.
\item \textbf{CollegeMsg}~\cite{Panzarasa2009}: private messages among students of an online
  community, the full record of $59\,835$ messages among $1899$ users. It contains $37$ exact
  duplicate rows; removing them changes $\hat\kappa$ by less than $0.1\%$.
\item \textbf{Stack Overflow}~\cite{Paranjape2017}: user--user interactions (answers and
  comments). To keep each user's timing and partner distribution intact we draw $400$ users at
  random among those with at least $10^3$ interactions and keep all of their outgoing
  interactions, whatever the partner ($1801$ nodes, $1.1\times10^{5}$ events).
\item \textbf{Wikibooks}~\cite{Kunegis2013}: edits of English Wikibooks pages, a bipartite
  editor--page log in which the editor is the initiator and the page the target; we use an
  active core of $6671$ nodes and $4.1\times10^{5}$ edits over $8.8$ years, built in the same
  way (Supplementary Material, Sec.~S3).
\end{itemize}

\paragraph*{Algorithmic records.}
\begin{itemize}
\item \textbf{Wikibooks bots}: edits of English Wikibooks made by bot accounts, $182$
  accounts and $26\,950$ edits on $11\,862$ pages between 2003 and 2010 (the period of the
  Wikibooks record), a machine-generated editing process without circadian or human
  constraints. Because the KONECT record has anonymised user identifiers, the bot record is
  built from the Wikimedia edit history, which carries usernames. An account is a bot if its
  name follows the Wikimedia bot naming convention (it contains ``bot'' as a word or suffix);
  $14$ of the $182$ also appear in the local bot-flag log. Accounts that held the bot flag but
  have human names (administrators running occasional mass edits) are excluded, and revisions
  imported from other wikis are dropped.
\item \textbf{Microservices}~\cite{Luo2021}: calls among $438$ services in the
  Alibaba 2021 microservice call-graph trace, the calling service being the initiator; a $2\times10^{5}$-call
  sample.
\end{itemize}
For these two records and the simulated network, the node-timing analyses use subsamples of
at most $6\times10^{4}$ events.

\paragraph*{Neural records.}
\begin{itemize}
\item \textbf{Simulated spiking network}: $40$ excitatory and $10$ inhibitory Izhikevich
  neurons~\cite{Izhikevich2003} with random connectivity (connection probability $0.2$),
  simulated for $10$~s at $1$~ms resolution. Each spike of neuron $i$ is recorded as one event
  $(t,i,j)$ for each of its postsynaptic targets $j$, so the initiator and the wiring are known
  exactly. We use a bursting and an asynchronous-irregular (refractory) regime.
\item \textbf{Cortex}~\cite{Siegle2021}: spike trains of single units recorded with Neuropixels
  probes in mouse visual cortex (Allen Brain Observatory, one session). For the node-timing
  statistics we use the $1546$ well-isolated units with at least $200$ spikes in the longest
  spontaneous (no-stimulus) block, $301$~s and $3.1\times10^{6}$ spikes, analysed in random sets
  of ten units. Extracellular recordings give no targets, so only the node-timing statistics are used and
  no $\kappa$ is fitted.
\end{itemize}

The contrast between the three SocioPatterns settings is a consistency check on the tie
channel: if $\hat\kappa$ were the same in a school, a conference and a hospital ward, it would be
measuring the instrument rather than a social mechanism. It is not (Sec.~\ref{sec:res_empirical}).

\subsubsection{Sessionisation}
\label{sec:sessionisation}

A SocioPatterns record is a \emph{sampled contact}, not an event: each line asserts that
two individuals were in proximity during a \samplingGrid{}~s window. A pair in
continuous conversation therefore emits a line every \samplingGrid{}~s. Fitting a
self-exciting process directly to these samples measures the sampling grid: the fitted
tie memory would recover $\gedge \approx 1/\samplingGrid{}\,\mathrm{s}^{-1}$ and a large
$\kappa$ as an artefact of the instrument, with no social content whatsoever.

We therefore \emph{sessionise}: consecutive samples of the same dyad separated by at
most $\Delta_{\mathrm{sess}} = \sessionGap{}$~s (two missed frames) are collapsed into a
single contact event, timed at the start of the run. The threshold is reported and its
sensitivity examined; the pipeline supports disabling sessionisation
($\Delta_{\mathrm{sess}}=0$) so that the magnitude of the artefact can be seen directly,
and we report that comparison in Sec.~\ref{sec:res_sensitivity}.

Sessionisation discards contact \emph{duration}, which is a real limitation. The present
model treats contacts as instantaneous marks; a duration-aware extension is natural but
changes the likelihood structure and is left to future work.

\subsubsection{Resolution cap and tied timestamps}

Even after sessionisation, distinct dyads share timestamps because the grid is common to
all pairs, and a node in a group conversation appears in several simultaneous contacts.
By Proposition~\ref{prop:divergence} this makes the timing likelihood unbounded, so the
decay rates are capped by Eq.~\eqref{eq:gammabounds} with $\Delta_{\min}$ equal to the
grid spacing. The number of tied events is reported for every fit.

\subsubsection{Daily segmentation}

Human contact data has a hard circadian structure with long overnight gaps containing no
events. A memory kernel fitted across such a gap spends its parameters explaining the
gap rather than the memory, biasing $\gnode$ and $\gedge$ towards the gap timescale. We
therefore fit each day separately, re-zeroing the time origin, and report the
between-day dispersion of the estimates rather than pooling them into a single number.

This dispersion is also our strongest available consistency check in the absence
of longitudinal re-observation: a stationary memory mechanism observed over multiple split-half
protocol should produce compatible mechanism estimates. Wide between-day scatter would
indicate that the model is absorbing day-specific structure it does not represent.

\subsubsection{Node filtering}

Nodes with fewer than \minContacts{} contacts are removed and the index set relabelled.
Such nodes contribute almost nothing to the likelihood but do enter the target-selection
normaliser $Z_i$ in Eq.~\eqref{eq:pi}, so a long tail of near-inactive nodes shifts
$\kappa$ downward by inflating the pool of available alternatives. This is a modelling
choice and is applied explicitly rather than silently; the sensitivity of the results to
the threshold is reported.

\subsubsection{A caveat on model adequacy}
\label{sec:adequacy}

Every estimate reported in Sec.~\ref{sec:results} is conditional on the model of
Sec.~\ref{sec:model} being an adequate description of the data. Parameter recovery on
synthetic data establishes only that the estimator inverts its own generative process.
A residual analysis based on the time-rescaling theorem (under which, for a correctly
specified conditional intensity, the compensated inter-event times are i.i.d.\ unit exponentials) is the
appropriate check, and we regard it as necessary before the empirical numbers can carry
mechanistic weight. Section~\ref{sec:res_gof} reports the rescaled-time Kolmogorov--Smirnov statistics with
parametric-bootstrap calibration for each dataset, and states plainly which settings the
model describes and which it does not.

\section*{Data availability}
The empirical datasets analysed in this study are publicly available in primary open-access repositories:
\begin{itemize}
\item Face-to-face proximity datasets (High School~\cite{Mastrandrea2015}, Conference SFHH~\cite{Genois2018}, and Hospital Ward~\cite{Vanhems2013}) are hosted by the SocioPatterns collaboration at \url{http://www.sociopatterns.org/datasets/}.
\item The baboon co-presence proximity dataset~\cite{Gelardi2020} is archived on Zenodo at \url{https://doi.org/10.5281/zenodo.3820287}.
\item Temporal communication networks (Email-Eu-core, CollegeMsg and Stack Overflow~\cite{Paranjape2017,Panzarasa2009}) are hosted by the Stanford Network Analysis Platform (SNAP) at \url{https://snap.stanford.edu/data/}.
\item The English Wikibooks edit log~\cite{Kunegis2013} is available from the KONECT collection (\texttt{edit-enwikibooks}) at \url{http://konect.cc/networks/}; the bot record is built from the Wikimedia stub-meta-history dump of English Wikibooks at \url{https://dumps.wikimedia.org/enwikibooks/}, with the script in the code repository.
\item The Alibaba microservice cluster traces~\cite{Luo2021} are available through GitHub at \url{https://github.com/alibaba/clusterdata} (DOI: 10.1145/3472883.3486987).
\item The cortical spike trains~\cite{Siegle2021} are from the Allen Brain Observatory Visual Coding Neuropixels dataset, available through the AllenSDK at \url{https://portal.brain-map.org/}.
\item The simulated spiking network is generated by a script in the code repository below.
\end{itemize}

\section*{Code availability}
All simulation routines, profile maximum-likelihood estimators, MC-EM latent-initiator algorithms, inter-event shuffle diagnostics, and epidemic spreading simulation scripts are openly available at \url{https://github.com/micheletizzani/temporal-memory-epidemics}. A self-contained reference implementation containing unit test suites, parity validation fixtures, and a one-step reproduction demonstration script is included directly within the repository.

\begin{acknowledgments}
M.T. thanks Sune Lehmann for intellectual freedom. This work was supported by the Independent Research Fund Denmark (EliteForsk grant to Sune Lehmann), the Carlsberg Foundation (The Hope Project), and the Villum Foundation (NationScale Social Networks).
\end{acknowledgments}

\section*{Author contributions}
M.T. designed the research, performed the analysis, and wrote the manuscript.

\section*{Competing interests}
The author declares no competing interests.

\bibliography{refs}

\end{document}